\documentclass[a4paper]{article}

\usepackage{ifthen}

\newboolean{lipics}
\newboolean{long}

\setboolean{lipics}{false}
\setboolean{long}{true}

\ifthenelse{\boolean{lipics}}{
  \newenvironment{compactenum}[1][1.]{\begin{enumerate}[#1]\setlength{\itemsep}{0pt}\setlength{\parskip}{0pt}}{\end{enumerate}}
  \newenvironment{compactitem}{\begin{itemize}\setlength{\itemsep}{0pt}\setlength{\parskip}{0pt}}{\end{itemize}}
}{
  \usepackage[utf8]{inputenc}
  \usepackage[T1]{fontenc}
  \usepackage[margin=1in]{geometry}
  \usepackage{amsthm}
  \usepackage{amsmath}
  \usepackage{amssymb}
  \usepackage{thmtools}
  \usepackage{thm-restate}
  \usepackage{hyperref}
  \usepackage{cleveref}
  \usepackage{subcaption}
  \usepackage{lmodern}
  \usepackage{paralist}
  \usepackage{authblk}
}

\usepackage{tikz}
\usepackage{xspace}

\makeatletter
\def\gnewcommand{\g@star@or@long\new@command}
\def\grenewcommand{\g@star@or@long\renew@command}
\def\g@star@or@long#1{\@ifstar{\let\l@ngrel@x\global#1}{\def\l@ngrel@x{\long\global}#1}}
\makeatother

\newcommand{\pqplan}{\textsc{Syn}\-\textsc{chro}\-\textsc{nized} \textsc{Pla}\-\textsc{nari}\-\textsc{ty}\xspace}
\newcommand{\cplan}{\textsc{Clus}\-\textsc{tered} \textsc{Pla}\-\textsc{nari}\-\textsc{ty}\xspace}
\newcommand{\atemb}{\textsc{Atom}\-\textsc{ic} \textsc{Em}\-\textsc{bed}\-\textsc{dabil}\-\textsc{i}\-\textsc{ty}\xspace}
\newcommand{\thicken}{\textsc{Thick}\-\textsc{en}\-\textsc{abil}\-\textsc{i}\-\textsc{ty}\xspace}
\newcommand{\ppqconp}{\textsc{Par}\-\textsc{tial}\-\textsc{ly} \textsc{PQ}-\textsc{con}\-\textsc{strained} \textsc{Pla}\-\textsc{nari}\-\textsc{ty}\xspace}
\newcommand{\sefe}{\textsc{SEFE}\xspace}
\newcommand{\consefe}{\textsc{Con}\-\textsc{nect}\-\textsc{ed} \textsc{SEFE}\xspace}

\newcommand{\convert}{\texttt{Con}\-\texttt{vert}\-\texttt{Small}\xspace}
\newcommand{\contract}{\texttt{En}\-\texttt{cap}\-\texttt{su}\-\texttt{late}\-\texttt{And}\-\texttt{Join}\xspace}
\newcommand{\propagate}{\texttt{Prop}\-\texttt{a}\-\texttt{gatePQ}\xspace}
\newcommand{\simplify}{\texttt{Sim}\-\texttt{pli}\-\texttt{fy}\-\texttt{Match}\-\texttt{ing}\xspace}

\newcommand{\pqplanTitle}{Synchronized Planarity\xspace}

\newcommand{\contractTitle}{EncapsulateAndJoin\xspace}
\newcommand{\propagateTitle}{PropagatePQ\xspace}
\newcommand{\simplifyTitle}{SimplifyMatching\xspace}

\newcommand{\convertop}{\convert\ensuremath{(u,I)}\xspace}
\newcommand{\contractop}{\contract\ensuremath{(\rho,I)}\xspace}
\newcommand{\propagateop}{\propagate\ensuremath{(u,I)}\xspace}
\newcommand{\simplifyop}{\simplify\ensuremath{(u,I)}\xspace}
\newcommand{\pqplanvars}[1][]{\ensuremath{I#1 = (G#1, \mathcal{P}#1, \mathcal{Q}#1, \psi#1)}\xspace}

\newcommand{\mysection}[1]{\smallskip\noindent\textbf{#1}}
\newlength{\punctuationfootlength}
\newcommand{\punctuationfootnote}[2]{#2\settowidth{\punctuationfootlength}{#2}\hspace{-0.5\punctuationfootlength}\nolinebreak\footnote{#1}}
\AtBeginDocument{\DeclareFontShape{T1}{lmr}{m}{scit}{<->ssub*lmr/m/scsl}{}}
\DeclareMathOperator{\vdeg}{deg}
\DeclareMathOperator{\vdegs}{deg*}

\ifthenelse{\boolean{lipics}}{}{
  \newtheorem{lemma}{Lemma}
  \newtheorem{theorem}[lemma]{Theorem}

}

\ifthenelse{\boolean{long}}{
  \newcommand{\proofomitted}{}
}{
  \newcommand{\proofomitted}{$\ast$}
  \let\origrestatable=\restatable
  \def\restatable{\origrestatable[\proofomitted]}
}

\crefname{case}{case}{cases}
\crefname{step}{step}{steps}

\graphicspath{{./graphics/}{./graphics/pdf/}}

\title{
  \pqplanTitle with Applications to Constrained Planarity Problems
  \ifthenelse{\boolean{lipics}}{}{
    \thanks{Work partially supported by DFG-grant Ru-1903/3-1.}
  }
}

\ifthenelse{\boolean{lipics}}{
  \author{Thomas Bläsius}{Faculty of Informatics, Karlsruhe Institute of Technology (KIT), Germany}{thomas.blaesius@kit.edu}{}{}

  \author{Simon D. Fink}{Faculty of Informatics and Mathematics, University of Passau, Germany}{finksim@fim.uni-passau.de}{https://orcid.org/0000-0002-2754-1195}{DFG-grant Ru-1903/3-1}

  \author{Ignaz Rutter}{Faculty of Informatics and Mathematics, University of Passau, Germany}{rutter@fim.uni-passau.de}{https://orcid.org/0000-0002-3794-4406}{DFG-grant Ru-1903/3-1}

  \authorrunning{T. Bläsius, S.\,D. Fink and I. Rutter} 

  \Copyright{Thomas Bläsius, Simon D. Fink and Ignaz Rutter}



  \category{}

\relatedversiondetails{Full Version}{https://arxiv.org/abs/2007.15362}

  \supplement{}

  \funding{Work partially supported by DFG-grant Ru-1903/3-1.}

  \nolinenumbers

\EventEditors{Petra Mutzel, Rasmus Pagh, and Grzegorz Herman}
  \EventNoEds{3}
  \EventLongTitle{29th Annual European Symposium on Algorithms (ESA 2021)}
  \EventShortTitle{ESA 2021}
  \EventAcronym{ESA}
  \EventYear{2021}
  \EventDate{September 6--8, 2021}
  \EventLocation{Lisbon, Portugal}
  \EventLogo{}
  \SeriesVolume{204}
  \ArticleNo{9}
}{
  \author[1]{Thomas Bl\"asius}
  \author[2]{Simon D. Fink}
  \author[2]{Ignaz Rutter}

  \affil[1]{Hasso Plattner Institute, University of Potsdam, Germany, \texttt{thomas.blaesius@hpi.de}}
  \affil[2]{Faculty of Informatics and Mathematics, University of Passau, Germany, \texttt{\{finksim,rutter\}@fim.uni-passau.de}}
}

\begin{document}

\ifthenelse{\boolean{lipics}}{}{
  \thispagestyle{empty}
  \pagenumbering{gobble}
}
\maketitle

\begin{abstract}
  We introduce the problem \pqplan. Roughly speaking, its input is a
  loop-free multi-graph together with synchronization constraints that,
  e.g., match pairs of vertices of equal degree by providing a bijection
  between their edges. \pqplan then asks whether the graph admits a
  crossing-free embedding into the plane such that the orders of edges
  around synchronized vertices are consistent. We show, on the one hand,
  that \pqplan can be solved in quadratic time, and, on the other hand,
  that it serves as a powerful modeling language that lets us easily
  formulate several constrained planarity problems as instances of
  \pqplan. In particular, this lets us solve \cplan in quadratic time,
  where the most efficient previously known algorithm has an upper bound
  of $O\left(n^{8}\right)$.
\end{abstract}

\newpage

\ifthenelse{\boolean{lipics}}{}{
  \pagenumbering{arabic}
}
\setcounter{page}{1}
\section{Introduction}
\label{sec:intro}

A graph is \emph{planar} if it admits an embedding into the plane that
has no edge crossings.  Planarity is a well-studied concept that
facilitates beautiful mathematical
structures~\cite{ht-dgtc-73,Booth1976}, allows for more efficient
algorithms~\cite{Hadlock1975}, and serves as a cornerstone in the
context of network visualization~\cite{Tamassia2013}. It
is not surprising that various generalizations, extensions, and
constrained variants of the \textsc{Planarity} problem have been
studied~\cite{Schaefer2013}.  Examples are \cplan, where the
embedding has to respect a laminar family of
clusters~\cite{Lengauer1989,Blaesius2015};
\textsc{Constrained Planarity}, where the orders of edges incident to
vertices are restricted, e.g., by PQ-trees~\cite{Blaesius2011}; and
\textsc{Simultaneous Planarity}, where two graphs sharing a
common subgraph must be embedded such that their embeddings coincide on
the shared part~\cite{Blaesius2017}.

For planar embeddings, there is the important notion of \emph{rotation}.
The rotation of a vertex is the counter-clockwise cyclic order of incident edges around it.  Many
of the above planarity variants come down to the question whether there
are embeddings of one or multiple graphs such that the rotations of
certain vertices are in sync in a certain way.  Inspired by this
observation, by the \atemb problem~\cite{Fulek2019}, and by the cluster
decomposition tree (CD-tree)~\cite{Blaesius2015}, we introduce a new
planarity variant.  \emph{\pqplan} has a loop-free multi-graph together
with two types of synchronization constraints as input.  Each
\emph{Q-constraint} is given as a subset of vertices together with a
fixed reference rotation for each of these vertices.  The Q-constraint
is satisfied if and only if either all these vertices have their
reference rotation or all these vertices have the reversed reference
rotation.  Vertices appearing in Q-constraints are called
\emph{Q-vertices} and all remaining vertices are
\emph{P-vertices}\punctuationfootnote{The names are based on PQ-trees,
where Q- and P-nodes have fixed and arbitrary rotation, respectively.}. 
A \emph{P-constraint} between two P-vertices $u$ and $v$ defines a
bijection between the edges incident to $u$ and $v$. It is satisfied if
and only if $u$ and $v$ have the opposite rotation under this bijection.
We require that the P-constraints form a matching, that is, no vertex
appears in more than one P-constraint.  The decision problem \pqplan now
asks whether the given graph can be embedded such that all Q- and all
P-constraints are satisfied.

\pqplan serves as a powerful modeling language that lets us express
various other planarity variants using simple linear-time reductions.
Specifically, we provide such reductions for \cplan, \atemb, \ppqconp,
and \textsc{Simultaneous Embedding with Fixed Edges} with a connected shared
graph (\consefe).  Our main contribution is an algorithm that solves
\pqplan, and thereby all the above problems, in quadratic time.

\subsection{Technical Contribution}
\label{sec:techn-contr}

Our result impacts different planarity variants that have been studied
previously.  Before discussing them individually in the context of
previous publications, we point out a common difficulty that has been
a major barrier for all of them, and briefly sketch how we resolve it.

Consider the following constraint on the rotation of a single vertex.
Assume its incident edges are grouped and we only allow orders where
no two groups alternate, that is, if $e_1,e_2$ are in one group and
$e_3,e_4$ are in a different group, then the circular subsequence
$e_1, e_3, e_2, e_4$ and its inverse are forbidden.  Such restrictions
have been called \emph{partition constraints}
before~\cite{Blaesius2015}, and they naturally emerge at cut-vertices
where each incident 2-connected component forms a group.  A single
partition constraint is not an issue by itself, but it becomes
difficult to deal with in combination with further restrictions.  This
is why cut-vertices and disconnected clusters are a major obstacle for
\sefe~\cite{Blaesius2017} and \cplan~\cite{Blaesius2015},
respectively.

The same issues appear for \textsc{Synchronized Planarity}, when we have a
cut-vertex $v$ that is involved in P-constraints, that is, its rotation
has to be synchronized with the rotation of a different vertex $u$.
We deal with these situations as follows, depending on whether $u$ is
also a cut-vertex or not.  If not, it is rather well understood which
embedding choices impact the rotation of $u$ and we can propagate this
from $u$ to $v$\punctuationfootnote{We can also do this if $v$ is not
  a cut-vertex.}.  This breaks the synchronization of $u$ and $v$ down
into the synchronization of smaller embedding choices.
This is a well-known technique that has been used
before~\cite{Blaesius2011,Carsten2008}.  If $u$ is also a
cut-vertex, we are forced to actually deal with the embedding choices
emerging at cut-vertices.  This is done by ``encapsulating'' the
restrictions on the rotations of $u$ and $v$ that are caused by the
fact that they are cut-vertices.  All additional restrictions coming
from embedding choices in the 2-connected components are pushed away
by introducing additional P-constraints.  After this, the cut-vertices
$u$ and $v$ have very simple structure, which can be resolved by
essentially joining them together.  This procedure is formally
described in Section~\ref{sec:contract} and illustrated in
Figures~\ref{fig:encapcut-p-matched} and~\ref{fig:contract-grid}.

This solution can be seen as combinatorial perspective on the recent
breakthrough result by Fulek and Tóth~\cite{Fulek2019}, who resolved
the cut-vertex issue by applying an idea coming from Carmesin's
work~\cite{Carmesin2017}.  While Carmesin works with 2-dimensional
simplicial complexes, Fulek and Tóth achieve their result by
transferring Carmesin's idea to the setting of topological graphs on
surfaces and combining it with tools from their work on
thickenability.  Our work transfers the problem and its solution back
to an entirely combinatorial treatment of topological graphs in the
plane.  This further simplification allows us to more clearly
highlight the key insight that makes the algorithm tick and at the
same time provides access to a wide range of algorithmic tools for
speeding up the computations.
\ifthenelse{\boolean{long}}{}{Due to space constraints, proofs of statements marked with a star are only given in the full version.}

\subsection{Related Work}
\label{sec:related-work}

\cplan was first considered by Lengauer~\cite{Lengauer1989} and later
rediscovered by Feng et al.~\cite{Feng1995}.  In both cases, the authors
give polynomial-time algorithms for the case that each cluster induces a
connected graph. The complexity of the general problem that allows
disconnected clusters has been open for 30 years.  In that time, many
special cases have been shown to be polynomial-time
solvable~\cite{Angelini2019,Cortese2008,Fulek2015,Gutwenger2002} before
Fulek and Tóth~\cite{Fulek2019} recently settled \cplan in P.  The core
ingredient for this is their $O(n^8)$ algorithm for the \atemb problem. 
It has two graphs $G$ and $H$ as input. Roughly speaking, $H$ describes
a 3-dimensional molecule structure with atoms represented by spheres and
connections (a.k.a.\ pipes) represented by cylinders.  The other graph
$G$ comes with a map to the molecule structure that maps each vertex to
an atom such that two neighboring vertices lie on the same atom or on
two atoms connected by a pipe. \atemb then asks whether $G$ can be
embedded onto the molecule structure such that no edges cross.

\atemb has been introduced as a generalization of the \thicken problem
that appears in computational topology~\cite{Akitaya2019}.
It can be shown that \atemb and \thicken \ifthenelse{\boolean{long}}{(and actually also \pqplan, as
discussed in \Cref{sec:comparison})}{} are linear-time
equivalent~\cite{Fulek2019}.
Thus, the above $O(n^8)$ algorithm for \atemb also solves \thicken and
\pqplan.
In a preprint, Carmesin~\cite{Carmesin2017} proves a Kuratowski-style
characterization of \thicken, which he claims yields a quadratic
algorithm as a byproduct.  While it is believable that the running
time of his algorithm is polynomial, a detailed runtime analysis is
missing.  In light of this, we only compare our algorithm to 
the~$O(n^8)$-algorithm by Fulek and Tóth.  
\ifthenelse{\boolean{long}}
{A detailed comparison of their solution to \atemb and our solution to \pqplan is given in \Cref{sec:comparison}.}
{For a detailed comparison of their solution to \atemb and our solution to \pqplan, we refer to the full version.}

To finally solve \cplan, Fulek and Tóth~\cite{Fulek2019} use the
reduction of Cortese and Patrignani~\cite{Cortese2018} to
\textsc{Independent Flat Clustered Planarity}, which they then reduce
further to \thicken. The last reduction to \thicken is based on a
combinatorial characterization of \thicken by Neuwirth~\cite{n-acmc-68},
which basically states that multiple graphs have to be embedded
consistently, that is, such that the rotation is synchronized between
certain vertex pairs of different graphs.
Via the reduction from \consefe to \cplan given by Angelini and Da
Lozzo~\cite{ad-scp-16}, the above result extends to \consefe, which was
a major open problem in the context of simultaneous graph
representations~\cite{bkr-hgdv-13}.  We flatten this chain of reductions
by giving a simple linear reduction from each of the problems \consefe,
\cplan, and \atemb to \ifthenelse{\boolean{long}}{\pqplan in \Cref{sec:applications},}{\pqplan,} yielding
quadratic-time algorithms for all of them.
\ifthenelse{\boolean{long}}{}{Due to space constraints, the reductions are only given in the full version.}
Moreover, the problem \ppqconp, for which we also give a linear
reduction to \pqplan, has been solved in polynomial time before, but
only for biconnected graphs~\cite{Blaesius2011} and in the non-partial
setting where either all or none of the edges of a vertex are
constrained~\cite{Carsten2008}.

\section{Preliminaries}\label{sec:preliminaries}

A \emph{partition} of a base set $X$ is a grouping of its elements
into non-empty subsets, the \emph{cells}, so that every element is
contained in exactly one cell. We assume a set implementation allowing
constant-time insertion and removal of elements, such as doubly-linked
lists with pointers stored with the elements.
When referring to graphs, we generally mean \emph{loop-free multi-graphs}.
A \emph{(multi-)star} consists of a \emph{center vertex} connected by multiple, possibly parallel, edges to its \emph{ray vertices}.
A \emph{$k$-wheel} is a $k$-cycle, where each node is also connected to an additional central node.
Furthermore, we assume a graph representation that allows efficient manipulation, such as an adjacency list with doubly-linked lists.

\mysection{Drawings, Embeddings and Cyclic Orders.}
A \emph{(topological) drawing} $\Gamma$ of a graph is a mapping of
every vertex $v$ to a point $p_v \in \mathbb{R}^2$ in the plane and a
mapping of every edge $\{u, v\}$ to a Jordan arc having $p_u$ and
$p_v$ as endpoints.
A drawing uniquely defines cyclic orders of edges incident to the same
vertex. Drawings with the same cyclic orders are considered
equivalent, their equivalence class is called \emph{(combinatorial)
embedding}. For an embedding~$\mathcal{E}$, we use $\mathcal{E}(u)$ to
denote the cyclic order of the edges incident to $u$ as given by
$\mathcal{E}$, which is also called the \emph{rotation} of $u$. For a
(cyclic) order $\sigma=\langle x_1,\ldots,x_k \rangle$ of $k$
elements, we use $\overline{\sigma}=\langle x_k,\ldots,x_1 \rangle$ to
denote its reversal.

\mysection{The \pqplanTitle Problem.}
An instance is a tuple $\pqplanvars$, where
\begin{compactenum}
\item $G=(P \cup Q, E)$ is a (loop-free) multi-graph with P-vertices $P$ and Q-vertices $Q$,
\item $\mathcal Q$ is a partition of~$Q$,
\item $\mathcal \psi$ is a mapping that assigns a rotation to each
  Q-vertex, and
\item $\mathcal P$ is a set of triples $(u,v,\varphi_{uv})$,
  where $u$ and $v$ are P-vertices of the same degree, $\varphi_{uv}$ is
  a bijection between their incident edges, and each P-vertex occurs at most once in~$\mathcal P$.
\end{compactenum}
We call the triples $\rho=(u,v,\varphi_{uv})$ in $\mathcal P$
\emph{pipes}. Pipes are not directed and we identify
$(u,v,\varphi_{uv})$ and $(v,u,\varphi_{vu})$ with
$\varphi_{vu}=\varphi_{uv}^{-1}$. We also define $\vdeg(\rho) =
\vdeg(u) = \vdeg(v)$. If two P-vertices are connected by a pipe, we call
them \emph{matched}; all other P- and Q-vertices are \emph{unmatched}.

The planar embedding $\mathcal E$ of $G$ \emph{satisfies the cell} $X
\in \mathcal Q$ if it is either $\mathcal E(v) = \psi(v)$ for all $v
\in X$ or $\mathcal E(v) = \overline{\psi(v)}$ for all $v \in X$.  We
say that the embedding satisfies the \emph{Q-constraints} if it
satisfies all cells, that is, vertices in the same cell of the
partition $\mathcal Q$ are consistently oriented. The embedding
$\mathcal E$ \emph{satisfies the pipe} $\rho=(u,v,\varphi_{uv})$ if
$\varphi_{uv}(\mathcal E(u)) = \overline{\mathcal E(v)}$, that is,
they have opposite rotations under the bijection~$\varphi_{uv}$. We
say that the embedding satisfies the \emph{P-constraints} if it
satisfies all pipes. The embedding~$\mathcal E$ is called \emph{valid}
if it satisfies the P-constraints and the Q-constraints.
The problem \pqplan asks whether a given instance $\pqplanvars$
admits a valid embedding.

\mysection{PQ-Trees and Embedding Trees.}
A \emph{PQ-tree} represents a set of circular orders of its leaves by
partitioning its inner nodes into two classes: For \emph{Q-nodes} the
rotation of incident edges is fixed up to reversal, for \emph{P-nodes},
this order can be chosen arbitrarily. Rooted PQ-trees have initially
been studied by Booth and Lueker \cite{Booth1976}. There is an
equivalence between rooted and unrooted PQ-trees \cite{Hsu2001}, where
the latter are also called PC-trees \cite{Shih1999}. We thus do not
distinguish them and simply use the term PQ-trees.
Note that a P-node with three or less neighbors allows the same
permutations as a Q-node of the same degree. We thus assume P-nodes to
have degree at least 4. We consider a PQ-tree \emph{trivial} if it
consists of a single inner P-node (with at least four leaves).
Otherwise, it consists of a single Q-node with at least two leaves, or
it contains at least two inner nodes, all of which have degree at least~3.

For a vertex of a planar biconnected graph, all rotations induced by planar
embeddings can efficiently be represented by a PQ-tree~\cite{Booth1976}.
This PQ-tree is also called the \emph{embedding tree} of the respective
node.
In the context of \pqplan, we assume that the embedding tree of a
vertex does not allow rotations that would result in a Q-vertex $v$ having
any other rotation than its default ordering~$\psi(v)$ or its
reverse~$\overline{\psi(v)}$.  To ensure this, we can subdivide each
edge incident to $v$ and connect each pair of two of the new nodes if
the edges they subdivide are consecutive in the cyclic order
$\psi(v)$~\cite{Carsten2008}. Note that this generates a $k$-wheel
with center $v$ and that there are exactly two planar rotations of the
center of a wheel, which are the reverse of each other.  We always
generate the embedding trees based on the graph where each Q-vertex in
$G$ is temporarily replaced with its respective wheel.

\mysection{Connected Components.}
A separating $k$-set is a set of $k$ vertices whose removal increases
the number of connected components. Separating 1-sets are called
\emph{cut-vertices}, while separating 2-sets are called \emph{separation
pairs}. A connected graph is \emph{biconnected} if it does not have a
cut-vertex. A biconnected graph is \emph{triconnected} if it does not
have a separation pair. Maximal biconnected subgraphs are called
\emph{blocks}.
A vertex that is not a cut-vertex and thus resides within an unique
block is called \emph{block-vertex}.

Hopcroft and Tarjan~\cite{ht-dgtc-73} define a graph decomposition into
triconnected components, also called
\emph{SPQR-tree}~\cite{Battista1996}, where the components come in three
shapes: \emph{bonds} consist of two \emph{pole} vertices
connected by multiple parallel edges, \emph{polygons} consist of a
simple cycle, and \emph{rigids}, whose embeddings are unique up to
reflection.  Each edge of these components is either \emph{real},
representing a single edge of the original graph, or \emph{virtual},
representing a subgraph.

Every planar embedding of a biconnected planar graph can be obtained
from an arbitrary planar embedding by flipping its rigids and reordering
the parallel edges in its bonds~\cite{ht-dgtc-73}.  The decomposition
can be computed in linear time~\cite{gm-ltisp-00} and can be used 
to compute the embedding
trees in linear time~\cite[Section 2.5]{Blaesius2011}.

\mysection{Splits and Joins of Graphs and Embeddings.}\label{sec:split-join}
Let $G=(V,E)$ be a graph. We call a partition  $C=(X,Y)$ of $V$ into two
disjoint cells a \emph{cut} of $G$. The edges $E(C)$ that have their
endpoints in different cells are called \emph{cut edges}.
The \emph{split} of $G$ at $C=(X,Y)$ is the disjoint union of the two
graphs obtained by contracting $X$ and $Y$ to a single vertex $x$ and
$y$, respectively (keeping possible multi-edges); see~\Cref{fig:join}.
Note that the edges incident to $x$ and $y$ are exactly the cut edges,
yielding a natural bijection~$\varphi_{xy}$ between them.
Conversely, given two graphs
$G_1=(V_1,E_1),G_2=(V_2,E_2)$ and vertices $x \in V_1$, $y \in V_2$
together with a bijection~$\varphi_{xy}$ between their incident edges,
their \emph{join} along~$\varphi_{xy}$ is the graph~$G=(V,E)$,
where $V=V_1\cup V_2\setminus\{x,y\}$ and $E$ contains all edges of $E_1
\cup E_2$ that are not incident to $x$ or $y$, and for each edge
$e=ux$ incident to $x$, $E$ contains an edge $uv$, where $v$ is the
endpoint of~$\varphi_{xy}(e)$ distinct from~$y$; see~\Cref{fig:join}. 
Observe that split and join are inverse operations.

We say that a planar embedding $\mathcal E$ of a graph $G$
\emph{respects} a cut $C=(X,Y)$ if and only if for a topological planar
drawing~$\Gamma$ of $G$ with embedding $\mathcal E$ there exists a
closed curve~$\gamma$ such that (i) $\gamma$ separates $X$ from~$Y$,
(ii) $\gamma$ crosses each edge in $E(C)$ in exactly one point, and
(iii) $\gamma$ does not cross any edge in $E \setminus E(C)$;
see~\Cref{fig:join}.  We say that $\gamma$ \emph{represents $C$
in~$\Gamma$}.

If $\mathcal E$ respects $C$, a split at $C$ preserves $\mathcal E$ as follows.
Let $G_1$ and~$G_2$ be the graphs resulting from
splitting $G$ at $C$ and let~$x \in V_1$ and $y \in V_2$ such
that~$\varphi_{xy}$ identifies their incident edges.  Let~$\Gamma$ be a
topological planar drawing with embedding~$\mathcal E$ and let~$\gamma$
be a curve in~$\Gamma$ that represents $C$ in~$\Gamma$.  We obtain
planar drawings~$\Gamma_1$ and~$\Gamma_2$ of~$G_1$ and~$G_2$ by
contracting to a single point the side of~$\gamma$ that contains $V_2$
and $V_1$, respectively.  We denote by~$\mathcal E_1$ and~$\mathcal E_2$
the corresponding combinatorial embeddings of $G_1$ and~$G_2$.  Note
that by construction for each vertex of~$V_1 \setminus \{x\}$ the
rotations in~$\mathcal E$ and~$\mathcal E_1$ coincide, and the same holds
for vertices of $V_2 \setminus \{y\}$ in $\mathcal E$ and $\mathcal
E_2$. Moreover, the rotations~$\mathcal E_1(x)$ and~$\mathcal E_2(y)$
are determined by the order in which the edges of $E(C)$ cross~$\gamma$,
and therefore they are oppositely oriented, that is, $\varphi_{xy}(\mathcal
E_1(x)) = \overline{\mathcal E_2(y)}$.  We call embeddings~$\mathcal
E_1$ and~$\mathcal E_2$ with this property \emph{compatible with
$\varphi_{xy}$}.

Conversely, we can join arbitrary embeddings $\mathcal E_1$ of~$G_1$ and
$\mathcal E_2$ of $G_2$ that are compatible with
$\varphi_{xy}$ by assuming that $x$ and~$y$ lie on the outer face,
removing $x$ and~$y$ from the embeddings, and connecting the resulting
half-edges according to~$\varphi_{xy}$.  The result is a planar
embedding~$\mathcal E$ where for each vertex $v \in V_i\setminus\{x,y\}$
it is~$\mathcal E(v) = \mathcal E_i(v)$ for $i=1,2$.

\begin{figure}[t]
  \centering
  \includegraphics[scale=1,page=4]{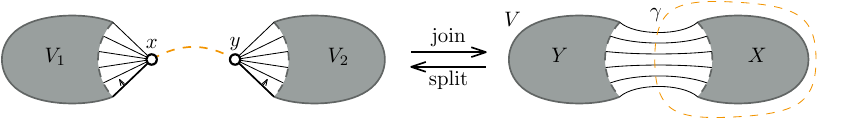}
  \caption{Joining and splitting two graphs at $x \in V_1$ and $y
  \in V_2$. The bijection $\varphi_{xy}$ between their incident
  edges is shown as follows: the two bold edges at the bottom are
  mapped to each other. The other edges are mapped according to
  their order following the arrow upwards (i.e. clockwise for $x$ and
  counter-clockwise for $y$).}
  \label{fig:join}
\end{figure}

\begin{restatable}{lemma}{lemSplitConnected}
  \label{lem:biconnected-all-embeddings-respect-cut}
  Let $G=(V,E)$ be a planar graph and let~$(X,Y)$ be a cut of $G$
  such that $X$ and~$Y$ induce connected subgraphs of $G$.  Then
  every planar embedding of~$G$ respects $(X,Y)$.
\end{restatable}
\newcommand{\proofSplitConnected}{
\begin{proof}
  Let~$\mathcal E$ be a planar embedding of $G$.  Since~$X$ and~$Y$
  induce connected subgraphs, it follows that no proper subset of
  $E(X,Y)$ is a cut.  Therefore, $(X,Y)$ is a so-called bond, which
  corresponds to a simple cycle $C^\star$ in the dual graph $G^\star$
  with respect to~$\mathcal E$~\cite[Proposition 4.6.1]{d-gt-17},
  which in turn implies that $\mathcal E$ respects $(X,Y)$.
\end{proof}}
\ifthenelse{\boolean{long}}{\proofSplitConnected}{\vspace{-0.2cm}}
\begin{restatable}{lemma}{lemSplitBipartite}
  \label{lem:bipartite-all-embeddings-respect-cut}
  Every planar embedding of a bipartite graph $G=(A\cup B, E)$
  respects $(A,B)$.
\end{restatable}
\newcommand{\proofSplitBipartite}{
\begin{proof}
  Let~$\mathcal E$ be a planar embedding of~$G$.  We claim that we can
  augment~$G$ and its embedding~$\mathcal E$ by additional edges in
  $\binom{A}{2} \cup \binom{B}{2}$ to a graph $G'$ with planar
  embedding~$\mathcal E'$ such that $A$ and $B$ induce connected
  subgraphs.  The fact that~$\mathcal E$ respects the cut $(A,B)$ then
  follows from Lemma~\ref{lem:biconnected-all-embeddings-respect-cut}.

  We prove the existence of such an augmentation by induction on
  $|A|+|B|$.  The statement clearly holds if~$|A|=|B|=1$.  Therefore
  assume~$|A|+|B|>2$ and without loss of generality assume $|A|>2$.
  Then there is a face that contains at least two vertices $a_1,a_2$
  of~$A$.  Let $G'=(A' \cup B',E')$ be the bipartite planar graph with
  planar embedding~$\mathcal E'$ resulting from adding the edges
  $a_1a_2$ and contracting it into a single vertex $a$.  By the
  inductive hypothesis, graph~$G'$ with embedding~$\mathcal E'$ can be
  augmented to bipartite $G''$ with embedding~$\mathcal E''$ so
  that~$A'$ and~$B'$ induce connected subgraphs.  We now expand $a$
  into two separate vertices $a_1,a_2$, undoing the previous
  contraction.  The resulting graph $G'''$ (including the edge
  $a_1a_2$) with embedding~$\mathcal E'''$ is the desired
  augmentation.
\end{proof}}
\ifthenelse{\boolean{long}}{\proofSplitBipartite}{}

\section{The \pqplanTitle Problem}\label{sec:pqplan}

We give an algorithm for solving \pqplan for graphs with  $n$ vertices
and $m$ edges in $O(m^2)$ time. Without loss of generality, we assume
that $G$ has no isolated vertices and thus~$m\in \Omega(n)$.
Furthermore, we assume the input graph $G$ to be planar.

\subsection{High-Level Algorithm}
\label{sec:high-level-algorithm}

Our approach hinges on three main ingredients.
The first are the three operations \contract, \propagate,
and \simplify, each of which can be applied to pipes that satisfy
certain conditions.  If an operation is applicable, it produces
an equivalent instance $I'$ of \pqplan in linear time. Secondly we show
that if none of the operations is applicable, then $I$ has no pipes,
and we give a simple linear-time algorithm for computing a valid
embedding in this case.  The third ingredient is a non-negative
potential function $\phi$ for instances of \pqplan.  We show that it
is upper-bounded by $2m$, and that each of the three operations
decreases it by at least~1.

Our algorithm is therefore extremely simple; namely, while the instance
still has a pipe, apply one of the operations to decrease the potential.
Since the potential function is initially bounded by $2m$, at most $2m$
operations are applied, each taking $O(m)$ time. We will show that
the resulting instance without pipes has size $O(m^2)$ and can be solved
in linear time, thus the total running time is $O(m^2)$.

\ifthenelse{\boolean{long}}{
  \subsubsection{Conversion of small-degree P-vertices}\label{deg-lt-3-to-q}
}{
  \mysection{Conversion of small-degree P-vertices.}\label{deg-lt-3-to-q}
}
The main difficulty in \pqplan stems from matched P-vertices. However,
P-vertices of degree up to~3 behave like Q-vertices in the sense that their
rotations are unique up to reversal.  Throughout this paper, we implicitly
assume that P-vertices of degree less than 4 are converted into Q-vertices,
also converting a pipe of degree less than 4 into a Q-constraint\ifthenelse{\boolean{long}}
{, using the auxiliary operation \convert described in the following.}
{; see the full version for additional details.}
We therefore assume without loss of generality that P-vertices, and in particular
pipes, have degree at least~4.

\newcommand{\secConvert}{
Vertices of degree~3 have only two distinct rotations, which are the
reverse of each other. Vertices of degree less than~3 have a unique
rotation, which coincides with its reverse. We thus define operation
\convertop to convert a P-vertex $u$ with $\deg(u) < 4$ into a
Q-vertex, resulting in an instance $\pqplanvars[']$ where $G'=(P' \cup
Q', E)$. If~$u$ is unmatched, we set $P'=P\setminus\{u\}$, $\mathcal
P'=\mathcal P$, and $Q'=Q\cup\{u\}$ and give $u$ it its own cell in
$\mathcal Q'=\mathcal Q \cup \left\{\{u\}\right\}$. We fix an
arbitrary order $\psi'(u)$ and let $\psi'$ coincide with $\psi$ for
all other vertices. If the P-vertex~$u$ is matched with another
P-vertex $v$, we can convert both of them to Q-vertices setting
$P'=P\setminus\{u,v\},$ $Q'=Q\cup\{u,v\}$, and $\mathcal P'=\mathcal
P\setminus\{(u,v,\varphi_{uv})\}$. We also put both together in a cell
in $\mathcal Q'= \mathcal Q \cup \left\{\{u,v\}\right\}$, again
setting $\psi'(u)$ as before, but now also defining $\psi'(v) =
\overline{\varphi_{uv}(\psi'(u))}$. Note that this enforces that
matched vertices $u$ and $v$ have opposite rotations under the
bijection~$\varphi_{uv}$. All other P-vertices and their pipes remain
unaffected. Previous Q-vertices already in $Q$ are also unaffected.

\begin{lemma}\label{lem:convert-small-corr}
Applying \convert to a P-vertex $u$ with $\deg(u) < 4$ yields an equivalent instance in constant time.
\end{lemma}
\begin{proof}
First we show that the conversion preserves a valid embedding
$\mathcal{E}$ of $I$. It is $\mathcal{E}(u)=\psi'(u)$ or
$\mathcal{E}(u)=\overline{\psi'(u)}$. If $u$ is unmatched, it is the
only Q-vertex in its cell and thus $\mathcal{E}$ also satisfies
the Q-constraints of $I'$. Otherwise, $u$ is matched with P-vertex
$v$ and as $\mathcal E$ satisfies the pipe $uv$ it is
$\varphi_{uv}(\mathcal E(u)) = \overline{\mathcal E(v)}$. If
$\mathcal{E}(u)=\psi'(u)$, we get $\psi'(v) =
\overline{\varphi_{uv}(\psi'(u))} =
\overline{\varphi_{uv}(\mathcal{E}(u))} = \mathcal{E}(v)$, satisfying
the new Q-constraint. The case of $\mathcal{E}(u)=\overline{\psi'(u)}$
follows analogously. As the underlying graph and all other pipes
remain unchanged, $\mathcal{E}$ is valid embedding of $I'$.

Conversely, assume that $\mathcal{E}'$ is a valid embedding for $I'$.
If $u$ is the sole vertex in its cell, converting $u$ to a
P-vertex will not affect the validity of the embedding (and also not
allow new embeddings as $\deg(u)<4$). If $u$ shares its cell with
vertex $v$, it is $\mathcal{E}'(u)=\psi'(u)$ and
$\mathcal{E}'(v)=\psi'(v)$ or $\mathcal{E}'(u)=\overline{\psi'(u)}$
and $\mathcal{E}'(v)=\overline{\psi'(v)}$. As we chose $\psi'(v) =
\overline{\varphi_{uv}(\psi'(u))}$, inserting shows that both cases
satisfy the constraint $\varphi_{uv}(\mathcal E(u)) =
\overline{\mathcal E(v)}$ of pipe $uv$. As the underlying graph, all
other pipes and Q-constraints remain unchanged, $\mathcal{E}'$ is
valid embedding of $I$.

This concludes the proof of correctness of \convert. As all affected
vertices have degree at most 3, the time required to execute the
operation is constant.
\end{proof}}
\ifthenelse{\boolean{long}}{
  \secConvert
}{}

\subsection{The \contractTitle Operation}\label{sec:contract}

The purpose of the \contract operation is to communicate embedding
restrictions between two matched cut-vertices in
two steps: First we encapsulate the cut-vertices into their own
independent star components, also disconnecting their incident
blocks from each other. In the second step, we join the stars.
\Cref{fig:encapcut-p-matched,fig:contract-grid} show an example.

\begin{figure}[t]
  \begin{subfigure}[t]{0.35\textwidth}
    \centering
    \includegraphics[page=1]{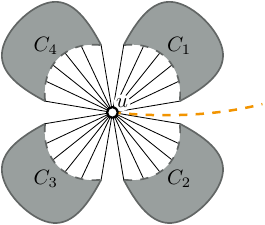}
    \vspace{-0.2cm}
    \caption{}
    \label{fig:encapcut-p-matched-pre}
  \end{subfigure}
  \hfill
  \begin{subfigure}[t]{0.6\textwidth}
    \centering
    \includegraphics[page=2]{graphics/encapcut.pdf}
    \vspace{-0.2cm}
    \caption{}
    \label{fig:encapcut-p-matched-post}
  \end{subfigure}
  \caption{A matched cut-vertex (a) and the result of encapsulating 
    it (b).}
  \label{fig:encapcut-p-matched}
\end{figure}

\begin{figure}[t]
  \begin{subfigure}[t]{0.44\textwidth}
    \includegraphics[scale=1,page=3]{graphics/grid.pdf}
    \vspace{-0.1cm}
    \caption{}
    \label{fig:contract-grid-pre}
  \end{subfigure}
  \hfill
  \begin{subfigure}[t]{0.44\textwidth}
    \includegraphics[scale=1,page=2]{graphics/grid.pdf}
    \vspace{-0.1cm}
    \caption{}
    \label{fig:contract-grid-post}
  \end{subfigure}
  \hfill
  \begin{subfigure}[t]{0.1\textwidth}
    \includegraphics[scale=1,page=1]{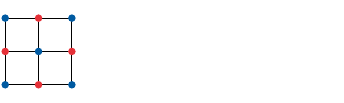}\vspace{-0.1cm}\caption{}\label{fig:contract-grid-res}
  \end{subfigure}
  \caption{Two (encapsulated) matched cut-vertices (a). Depending on
    the mapping $\varphi$, any bipartite graph can result from
    joining them. For example, the graph (b) can result,
    which is isomorphic to the square grid graph shown in (c).}
  \label{fig:contract-grid}
\end{figure}

For an instance \pqplanvars of \pqplan, let $\rho=(u,v,\varphi_{uv})$ be
a pipe matching two cut-vertices $u,v$ of two (not necessarily distinct)
connected components $C_u,C_v$ of $G$. Operation \contractop can be applied
resulting in an instance \pqplanvars['] using the following two steps.
We first preprocess both cut-vertices to \emph{encapsulate} them into
their own separate star components. Let~$C_1,\dots,C_k$ be the
connected components of $C_u-u$.  We split $C_u$ along the cuts
$(V(C_i),V \setminus V(C_i))$ for $i=1,\dots,k$.  We denote the vertices
resulting from the split along $(V(C_i), V \setminus V(C_i))$ as $u_i$
and~$u_i'$, where~$u_i$ results from contracting~$V \setminus V(C_i)$
and~$u_i'$ results from contracting~$V(C_i)$.  Note that, after all
splits, $u$ is the center of a star~$C_u'$ whose ray vertices are the $u_i'$. 
We add the pipes~$(u_i,u_i',\varphi_{u_iu_i'})$ for $i=1,\dots,k$; see
\Cref{fig:encapcut-p-matched}.  The same procedure is also applied to
$v$, resulting in an intermediate instance $I^\ast$.
In the second step, we \emph{join} the connected components~$C_u'$
and~$C_v'$ at~$u$ and~$v$ along the mapping~$\varphi_{uv}$ of $\rho$
into a component~$C_{uv}$.  We also remove the
pipe~$\rho$ from $I^\ast$; all other parts of the instance remain
unchanged.  \Cref{fig:contract-grid} shows a possible result of joining
two stars.

\begin{lemma}\label{lem:contract-corr}
  Applying \contract to a pipe $\rho$ yields an equivalent instance in
  $O(\vdeg(\rho))$ time.
\end{lemma}
\begin{proof}
  By~\Cref{lem:biconnected-all-embeddings-respect-cut}, a valid embedding
  $\mathcal E$ of an instance $I$ respects each of the cuts $(V(C_i),V \setminus
  V(C_i))$ for $i=1,\dots,k$, yielding a planar embedding~$\mathcal
  E^\ast$ of $I^\ast$.  By construction, it is~$\mathcal E^\ast(u_i) =
  \overline{\mathcal E^\ast(u_i')}$ for $i=1,\dots,k$, that is, each new
  pipe~$(u_i,u_i',\varphi_{u_iu_i'})$ is satisfied and~$\mathcal E^\ast$
  is a valid embedding of $I^\ast$.
Conversely, if~$\mathcal E^\ast$ is a valid embedding of $I^\ast$, we
  can join $u_i$ with $u_i'$ for $i=1,\dots,k$ to obtain a valid planar
  embedding~$\mathcal E$ of $I$, as the pipe
  $(u_i,u_i',\varphi_{u_iu_i'})$ ensures that $\mathcal E^\ast$ is
  compatible with $\varphi_{u_iu_i'}$.
The same applies to $C_v$.

  If~$\mathcal E^\ast$ is a valid embedding for $I^\ast$, it satisfies
  the pipe $(u,v,\varphi_{uv})$ and we can join the embedding at $u$
  and~$v$ via~$\varphi_{uv}$ to obtain a planar embedding $\mathcal E'$
  of~$G'$.  Since the rotations of vertices different from $u,v$ are
  unaffected, $\mathcal E'$ is valid for $I'$. Conversely, assume that
  $\mathcal E'$ is a valid embedding for $I'$. Note that joining two
  stars at their centers yields a bipartite graph consisting of the rays
  of the former stars.  Thus~$C_{uv}$ is bipartite, and
  by~\Cref{lem:bipartite-all-embeddings-respect-cut} every embedding
  respects the cut of the bipartition. Thus, we can split $\mathcal E'$
  and obtain a valid embedding of $I^\ast$.
  
  As the operation affects exactly the edges incident to $u$ and $v$ and
  potentially creates a new structure with size proportional to their
  number, its running time is linear in the degree of the affected
  pipe.
\end{proof}

Observe that this operation replaces a pipe and two cut-vertices by
smaller pipes and smaller cut-vertices, respectively. Through multiple
applications of \contract we can thus step by step decrease the degree
of cut-vertex-to-cut-vertex pipes, until there are none left in the
instance.
Note that \contract can yield an arbitrary bipartite component.
If the component is non-planar, we abort and report a no-instance.

\subsection{The \propagateTitle Operation}\label{sec:propagate}

The operation \propagate communicates embedding restrictions of a
biconnected component across a pipe. These restrictions are represented
by the embedding tree of the matched P-vertex of interest.  Both
endpoints of the pipe are replaced by copies of this tree.  To ensure
that both copies are embedded in a compatible way, we synchronize their
inner nodes using pipes and Q-constraints; see \Cref{fig:propagate}.

\begin{figure}[t]
  \begin{subfigure}[t]{0.41\textwidth}
    \centering
    \includegraphics[page=1]{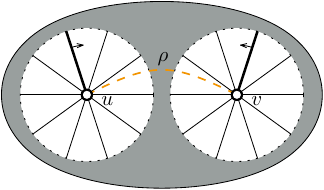}
    \vspace{-0.3cm}
    \caption{}
    \label{fig:propagate-pre}
  \end{subfigure}
  \begin{subfigure}[t]{0.59\textwidth}
    \centering
    \includegraphics[page=2]{graphics/propagate.pdf}
    \vspace{-0.3cm}
    \caption{}
    \label{fig:propagate-post}
  \end{subfigure}
  \caption{A block-vertex $u$ matched with vertex $v$ (a); the
    bijection $\varphi_{uv}$ maps the bold edge of $u$ to the bold edge
    of $v$, the remaining edges are mapped according to their
    order, clockwise around $u$ and counter-clockwise around
    $v$.  The result of applying \propagateop (b). Note that the
    second inserted tree $T_u'$ is mirrored with respect to
    $T_u$. Q-vertices and -nodes are drawn as squares while
    P-vertices and -nodes are drawn as disks.
  }
  \label{fig:propagate}
\end{figure}

For an instance \pqplanvars of \pqplan, let $u$ be a block-vertex
matched by a pipe $\rho=(u,v,\varphi_{uv})$.  If the
embedding tree~$\mathcal T_u$ of $u$ is non-trivial (i.e., it not only
consists of a single P-node), then the operation \propagateop can be
applied, resulting in an instance \pqplanvars['] as follows.
We turn the PQ-tree $\mathcal{T}_u$ into a tree $T_u$ by interpreting
Q-nodes as Q-vertices and P-nodes as P-vertices.  To construct $G'$ from
$G$, we replace~$u$ with~$T_u$ by reconnecting the incident edges of $u$
to the respective leaves of $T_u$.  We also replace~$v$ by a second copy
$T_u'$ of $T_u$ by reconnecting an edge $e$ incident to $v$ to the leaf
of $T_u'$ that corresponds to $\varphi_{vu}(e)$.
For a vertex $\alpha$ of $T_u$ we denote the corresponding vertex of
$T_u'$ by $\alpha'$. For an edge $\alpha\beta$ of $T_u$ we define
$\varphi_{T_uT_u'}(\alpha\beta)=\alpha'\beta'$.
For each Q-vertex $\alpha$ of $T_u$, we define $\psi'(\alpha)$
according to the rotation of the corresponding Q-node in~
$\mathcal{T}_u$.  For the Q-vertex $\alpha'$ of $T_u'$, we define
$\psi'(\alpha') = \overline{\varphi_{T_uT_u'}(\psi'(\alpha))}$. For
all other Q-vertices of~$I$, $\psi'$ coincides with $\psi$. We define
the partition $\mathcal Q' = \mathcal Q \cup \left\{  \{\alpha,
\alpha'\} \mid \alpha \text{ is a Q-vertex of } T_u \right\}$.
For each P-vertex $\alpha$ of $T_u$, we define a pipe $\rho_\alpha = 
(\alpha, \alpha', \varphi_{\alpha\alpha'})$ with
$\varphi_{\alpha\alpha'}(e)=\varphi_{T_uT_u'}(e)$ for each edge~$e$
incident to  $\alpha$. Finally, we define the matching $\mathcal P' = 
(\mathcal P \setminus \{\rho\}) \cup \left\{  \rho_\alpha \mid \alpha
\text{ is a P-vertex of } T_u \right\}$.

\begin{restatable}{lemma}{lemPropagate}
  Applying \propagate to a block-vertex $u$ with a non-trivial embedding
  tree yields an equivalent instance.
If the embedding tree $\mathcal{T}_u$ is known, 
  operation \propagate runs in $O(\vdeg(u))$ time.
\end{restatable}
\newcommand{\proofPropagate}{
\begin{proof}
First we show that \propagate preserves a valid embedding
$\mathcal{E}$ of $I$. To define an embedding $\mathcal{E}'$ for $I'$,
we substitute $u$ and $v$ in $\mathcal{E}$ by suitably embedded trees
$T_u$ and $T_u'$.
The tree inserted at $u$ represents all the possible rotations of $u$,
including $\mathcal{E}(u)$ and the insertion can therefore be done
without introducing crossings. As $\mathcal{E}$ fulfilled the
P-constraint of the pipe $uv$, we know that
$\mathcal{E}(u)=\overline{\varphi_{vu}(\mathcal{E}(v))}$. Therefore,
the same holds for inserting the mirrored copy $T_u'$ of $T_u$ instead
of $v$. Thus, the resulting embedding $\mathcal{E}'$ is planar. Note
that, the mirror embedding of $T_u$ is obtained by reversing the
rotation of each inner vertex of $T_u$. Therefore, for each inner
vertex $\alpha$ of $T_u$ it holds that
$\mathcal{E}'(\alpha)=\overline{\varphi_{\alpha'\alpha}(\mathcal{E}'(\alpha'))}$.
Thus, all the new pipes are satisfied.
Similarly, for each inner Q-vertex $\alpha$ of $T_u$ the rotation in
$\mathcal{E}'$ is either $\psi'(\alpha)$ or
$\overline{\psi'(\alpha)}$. As the rotation of $\alpha'$ in
$\mathcal{E}'$ is mirrored, the new Q-constraints are satisfied. Since
all other P- and Q-constraints remain satisfied, $\mathcal{E}'$ is a
valid embedding for $I'$.

Conversely, assume that $\mathcal{E}'$ is a valid embedding for $I'$.
We obtain an embedding $\mathcal{E}$ of $I$ by contracting $T_u$ and
$T_v$ into single vertices $u$ and $v$, respectively. Clearly,
$\mathcal{E}$ is a planar embedding. All Q-constraints of $I$ and also
all pipes except for $\rho$ clearly remain satisfied. It remains to
show that also pipe $\rho$ is satisfied. For each vertex $\alpha$ of
$T_u$ it holds that
$\mathcal{E}'(\alpha)=\overline{\varphi_{\alpha'\alpha}(\mathcal{E}'(\alpha'))}$,
as $\mathcal{E}'$ in particular fulfills all new Q- and P-constraints
of~$I'$. Therefore, the circular order of the leaves of~$T_u$ is the
reverse of the order of the leaves of $T_u'$. This means that
$\mathcal{E}(u)=\overline{\varphi_{vu}(\mathcal{E}(v))}$, that is,
$\rho$ is satisfied.

This concludes the proof of correctness of \propagate. As the
operation affects exactly the edges incident to $u$ and $v$, its
running time is linear in the degree of the affected node, given the
embedding tree $\mathcal T_u$ is known.
\end{proof}}
\ifthenelse{\boolean{long}}{\proofPropagate}{}

Note that the tree $T_u'$ inserted instead of $v$ may not be
compatible with the rotations of $v$. In this case, the component
becomes non-planar, potentially causing the later generation of an
embedding tree to fail.
\ifthenelse{\boolean{long}}{
To detect this early, we can also compute the
embedding tree $\mathcal{T}_v$ of $v$ and intersect $\mathcal{T}_u$
with $\mathcal{T}_v$ before the insertion.  Either way, if the
generation of an embedding tree or the intersection of two embedding
trees fails, we can report a no-instance.
}{We can then immediately report a no-instance.}

Observe that since we assume $\mathcal T_u$ to be non-trivial, the
degrees of all P-vertices in $T_u$ and $T_u'$ are strictly smaller than
the degree of $u$.  Thus, by repeatedly applying \propagate, we
eventually arrive at an equivalent instance where all matched
block-vertices have a trivial embedding tree.
Also note that if $\mathcal{T}_u$ consists of a single Q-node,
\propagate effectively replaces the affected pipe by two Q-vertices in
the same partition. The case where $\mathcal{T}_u$ is trivial and thus
consists of a single P-node is handled by the next operation.

\subsection{The \simplifyTitle Operation}\label{sec:simplify}

The remaining operation is \simplify, which is used to resolve pipes
where one side has no restrictions to be communicated to the other side.
This is the case when one of the two matched vertices is a pole of a
bond that allows arbitrary rotation. We distinguish three cases:
\ref{case:simplify-unmatched}) bonds where one pole can always mimic the
rotation of the other, \ref{case:simplify-toroidal}) bonds where the
pipe synchronizes one pole with the other (similar to the toroidal
instances of Fulek and Tóth~\cite{Fulek2019}), and
\ref{case:simplify-matched}) bonds that link two distinct pipes.

For an instance \pqplanvars of \pqplan, let~$u$ be a block-vertex of $G$
whose embedding tree is trivial and that is matched by a pipe $\rho$.
Then, its embedding is determined by exactly one triconnected component
$\mu$, which is a bond\punctuationfootnote{as a second bond would cause
another P-node in the embedding tree, a rigid would cause a Q-node and
polygons do not affect the embedding trees~\cite[Section
2.5]{Blaesius2011}}. Thus~$u$ is the pole of bond~$\mu$, and we call the
vertex $v$ that is the other pole of~$\mu$ the \emph{partner} of~$u$. If
$v$ is either unmatched or a block-vertex with a trivial embedding tree,
the operation \simplifyop can be applied, resulting in an instance
\pqplanvars['] as follows. Note that, due to the temporary replacement of
Q-vertices by wheels when computing the embedding trees, $v$ cannot be a
Q-vertex, as that would make the PQ-tree of $u$ contain a Q-node.

\begin{compactenum}[(i)]
\item \label[case]{case:simplify-unmatched} If $v$ is an unmatched
  P-vertex (\Cref{fig:simplify-unmatched}), $I'$ is obtained
  from $I$ by removing~$\rho$.

\item \label[case]{case:simplify-toroidal} If $\rho$ matches $u$ with
  $v$, it connects the two poles of the bond $\mu$ (\Cref{fig:simplify-toroidal}).
  Note that the embedding trees of $u$
  and $v$ both contain a P-node of the same degree representing $\mu$
  and the pipe now requires both $u$ and $v$ to have the same degree.
  Thus, as $u$ has a trivial embedding tree, $v$ also has a trivial
  embedding tree. The rotation of the vertices is thus
  exclusively determined by the embedding of the bond and there are
  bijections $\delta_u$ and $\delta_v$ between the edges incident to $u$
  and $v$, respectively, and the virtual edges within the bond. We now
  check that these bijections are compatible with the bijection 
  $\varphi_{uv}$ given by the pipe.
Let $\delta_{vu}=\delta_u^{-1} \circ \delta_v$ be a bijection between
  the edges incident to $v$ and the edges incident to $u$, and let $\pi
  = \varphi_{uv} \circ \delta_{vu}$ be a permutation of the edges
  incident to $v$.
  If all cycles of $\pi$ have the same length, $I'$ is obtained from $I$ by removing~$\rho$\punctuationfootnote{
    If all cycles of $\pi$ have the same length, $\pi$ is \emph{order preserving} and it is $\pi(O)=O$ for any sequence $O$.
    See \cite[Lemma 2.2]{Blaesius2011} or the proof to the following \Cref{lem:simplify} in the full version for more details.
  }.
  Otherwise, $I$ is an invalid instance and we set $I'$ to a trivial no-instance.

\item \label[case]{case:simplify-matched} If $v$ is matched with a
  P-vertex $v' \ne u$ via pipe $\rho'=(v,v',\varphi_{vv'})$, let $u'$ be
  the other endpoint of~$\rho=(u,u',\varphi_{uu'})$.  We remove $\rho$
  and $\rho'$ and add the new pipe $\rho^*=(u',v',\varphi_{u'v'})$ with 
  $\varphi_{u'v'} = \varphi_{vv'} \circ \delta_{uv} \circ
  \varphi_{u'u}$; see \Cref{fig:simplify-matched}.
\end{compactenum}

\begin{figure}[t]
  \begin{subfigure}[t]{0.34\textwidth}
    \centering
    \includegraphics[scale=1,page=2]{graphics/simplify2.pdf}\vspace{-0.2cm}\caption{}\label{fig:simplify-unmatched}
  \end{subfigure}\begin{subfigure}[t]{0.26\textwidth}
    \centering
    \includegraphics[trim=30 0 30 0,clip,scale=1,page=3]{graphics/simplify2.pdf}
    \vspace{-0.2cm}
    \caption{}
    \label{fig:simplify-toroidal}
  \end{subfigure}\begin{subfigure}[t]{0.40\textwidth}
    \centering
    \includegraphics[scale=1,page=1]{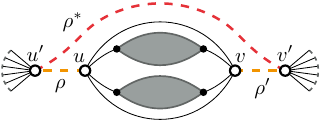}
    \vspace{-0.2cm}
    \caption{}
    \label{fig:simplify-matched}
  \end{subfigure}
  \caption{The three cases of the \simplifyTitle operation. In
    \Cref{case:simplify-unmatched} (a) and
    \Cref{case:simplify-toroidal} (b), the pipe~$\rho$ is
    removed. In \Cref{case:simplify-matched} (c) the
    pipes~$\rho,\rho'$ are replaced by pipe~$\rho^\star$.}
  \label{fig:simplify}
\end{figure}

\begin{restatable}{lemma}{lemSimplify}\label{lem:simplify}
  Applying \simplify to a block-vertex $u$ with a trivial embedding tree yields
  an equivalent instance in $O(\vdeg(u))$ time.

\end{restatable}
\newcommand{\proofSimplify}{
\begin{proof}
  First we show that if $\mathcal E$ is a valid embedding of $I$, then
  it is also a valid embedding of~$I'$.  This clearly holds if $I'$ is
  obtained from $I$ by removing a pipe, as in
  \Cref{case:simplify-unmatched}.

  It remains to investigate \Cref{case:simplify-toroidal} and \Cref{case:simplify-matched}.
  In the latter case, let~$u'$ be the vertex to which $u$ is matched
  and let~$v' \ne u$ be the vertex to which $v$ is matched.
  We want to show that~$\mathcal E$ satisfies the constraint
  $\mathcal{E}(v') = \overline{\varphi_{u'v'}(\mathcal{E}(u'))}$ of
  the newly added pipe $\rho^*=(u',v',\varphi_{u'v'})$.
  By assumption the pipes $\rho,\rho'$ are satisfied by $\mathcal{E}$,
  that is $\mathcal{E}(u) = \overline{\varphi_{u'u}(\mathcal{E}(u'))}$
  and $\mathcal{E}(v') = \overline{\varphi_{vv'}(\mathcal{E}(v))}$.
  Moreover, as both $u$ and $v$ have a trivial embedding tree, the
  bijections $\delta_u,\delta_v$ are defined as before.  Thus
  $\delta_{u}(\mathcal{E}(u))$ and $\delta_{v}(\mathcal{E}(v))$ define
  a circular orders of the virtual edges of $\mu$.  Since the
  embedding is planar,
  $\delta_{u}(\mathcal{E}(u))=\overline{\delta_{v}(\mathcal{E}(v))}$.
  This yields
  $\mathcal{E}(v) = \overline{\delta_{uv}(\mathcal{E}(u))}$ and thus
\begin{align*}
\mathcal{E}(v') 
&= \overline{\varphi_{vv'}(\overline{\delta_{uv}(                        \mathcal{E}(u)   )})} 
 =           \varphi_{vv'}(          \delta_{uv}(                        \mathcal{E}(u)   ) ) \\
&=           \varphi_{vv'}(          \delta_{uv}(\overline{\varphi_{u'u}(\mathcal{E}(u'))}) ) 
 = \overline{\varphi_{vv'} \circ     \delta_{uv} \circ     \varphi_{u'u}(\mathcal{E}(u'))    }
 = \overline{\varphi_{u'v'}(\mathcal{E}(u'))}\,.
\end{align*}
In particular, this means that $\rho^*$ is satisfied and therefore
$\mathcal{E}$ is valid.

In \Cref{case:simplify-toroidal}, we consider the above permutation
$\pi$.  Analogously to $\mathcal{E}(v) =
\overline{\delta_{uv}(\mathcal{E}(u))}$, it is $\mathcal{E}(u) =
\overline{\delta_{vu}(\mathcal{E}(v))}$, and since $\mathcal E$
satisfies $\rho$, we find that $\pi(\mathcal E(v)) = \varphi_{uv}
\circ \delta_{vu}(\mathcal E(v)) = \varphi_{uv}( \overline{\mathcal
E(u)} ) = \mathcal E(v)$.  All cycles of $\pi$ have the same
length~\cite[Lemma 2.2]{Blaesius2011} and therefore $I'$ is obtained
from $I$ by removing $\rho$ and, in particular, $\mathcal{E}$ is a
valid embedding of $I'$.

Conversely, assume that~$\mathcal{E}'$ is a valid embedding for~$I'$.
To obtain a valid embedding~$\mathcal{E}$ of~$I$, we modify the
embedding $\mathcal{E}'$ by changing the order of the virtual edges of
the bond~$\mu$ with poles~$u$ and~$v$ in such a way that the removed
pipes are satisfied.  Since, compared to~$\mathcal E'$, we only change
the embedding of a bond, $\mathcal E$ is guaranteed to be planar.
The details depend on which case of the operation applies.

If $v$ is unmatched in $I$, we change the embedding of $\mu$ such that
$\mathcal{E}(u)=\overline{\varphi_{u'u}(\mathcal{E}(u'))}$.  This is
possible since there is a bijection between the edges incident to $u$
and the virtual edges of $\mu$.  The new embedding $\mathcal E$
satisfies $\rho$ and, since $v$ is unmatched, also all other pipes of
$I$ remain satisfied.

In \Cref{case:simplify-toroidal}, $v$ is matched with $u$ and we
consider the above permutation $\pi$.  As $\mathcal E'$ is a valid
embedding, $I'$ results from $I$ by removing $\rho$ and we know that
all cycles of $\pi$ have the same length.  There exists a circular
ordering $\sigma$ of the edges incident to $v$ with
$\pi(\sigma)=\sigma$~\cite[Lemma 2.2]{Blaesius2011}. We change the
embedding of $\mu$ such that $\mathcal{E}(v)=\sigma$ and, to retain
planarity, $\mathcal{E}(u) = \overline{\delta_{vu}(\mathcal{E}(v))}$. 
This satisfies $\rho$ as $ \overline{\varphi_{uv}(\mathcal{E}(u))} =
\overline{\varphi_{uv}(\overline{\delta_{vu}(\mathcal{E}(v))}} =
{\varphi_{uv}( {\delta_{vu}(\mathcal{E}(v))}} = \pi(\mathcal{E}(v)) =
\mathcal{E}(v) $ and thus $\mathcal E$ is valid.

In \Cref{case:simplify-matched}, $v$ is matched with $v'\neq u$.  We
obtain $\mathcal E$ from $\mathcal E'$ by setting
$\mathcal{E}(u) = \overline{\varphi_{u'u}(\mathcal{E}(u'))}$ and
$\mathcal{E}(v) = \overline{\varphi_{v'v}(\mathcal{E}(v'))}$.  This
satisfies $\rho$ and $\rho'$, and to verify the planarity of
$\mathcal{E}$ it suffices to show that
$\mathcal{E}(v) = \overline{\delta_{uv}(\mathcal{E}(u))}$.  Since
$\rho^*$ is satisfied by $\mathcal E'$ and $\mathcal E$ differs from
$\mathcal E'$ only at $u$ and $v$, we know that
$\mathcal E(v') = \overline{\varphi_{u'v'}(\mathcal E(u'))} =
\overline{\varphi_{vv'} \circ \delta_{uv} \circ \varphi_{u'u}(\mathcal
  E(u'))}$.  This is equivalent to
$\varphi_{v'v}(\mathcal E(v')) =
\overline{\delta_{uv}(\varphi_{u'u}(\mathcal E(u')))}$.  By the
definitions of $\mathcal E(u)$ and $\mathcal E(v)$, this yields
$\mathcal{E}(v) = \overline{\delta_{uv}(\mathcal{E}(u))}$, that is
$\mathcal E$ is planar.

This concludes the proof of correctness of \simplify.  The updates to
the matching $\mathcal P$ can be done in constant time, while the
updates to the bijection $\varphi$ and the check for cycle lengths
require $O(\vdeg(u))$ time.
\end{proof}}
\ifthenelse{\boolean{long}}{\proofSimplify}{}

\subsection{Reduced and Pipe-Free Instances}

With our exposition of the fundamental operations complete, we now
study how to solve instances where none of those operations can be
applied.  We call such instances \emph{reduced}.

\begin{lemma}\label{lem:reduced-instance}
  An instance is reduced if and only if it contains no pipes.
\end{lemma}
\begin{proof}
  Obviously, a pipe-free instance is reduced.
  Conversely, consider a reduced instance $I$.  Assume, for the sake
  of contradiction, that $I$ contains a
  pipe. We now show that this implies that one of the operations is
  applicable, that is, $I$ is not reduced.

  Assume that $I$ contains no matched cut-vertices and thus all
  matched vertices are block-vertices.
If there is a matched P-vertex with a non-trivial embedding tree,
  \propagate can be applied.  Otherwise, all matched P-vertices are
  block-vertices with trivial embedding trees and \simplify can be applied.

  Now let $u$ be a matched cut-vertex of maximum degree that is
  matched to a vertex~$v$ by a pipe~$\rho$.  If~$v$ is also a
  cut-vertex, we can apply \contract.  If~$v$ is a block-vertex with a
  non-trivial embedding tree, we can apply \propagate.  Therefore, $v$
  must be a block-vertex with a trivial embedding tree.  Now we can apply
  \simplify, unless the partner pole $v'$ of $v$ is a matched cut-vertex.
  This is however excluded, since~$\deg(u) = \deg(v) < \deg(v')$,
  contradicting the maximality of $\deg(u)$.  The last inequality
  follows from the fact that~$\deg(v) \leq \deg(v')$ already holds in
  the block of $G$ that contains $v$ and~$v'$, but as~$v'$ is a
  cut-vertex, it has at least one neighbor outside that block.
\end{proof}

To solve instances without pipes in linear time, note that a planar
embedding of such an instance is valid if and only if it satisfies the
Q-constraints.  As Q-vertices only have a binary choice for their
rotation, it is relatively easy to synchronize them via a 2-SAT formula. 
Linear-time algorithms follow from, e.g., \cite{Blaesius2011}, and can
also be obtained from techniques similar to those used by Fulek and
Tóth~\cite{Fulek2019} for cubic graphs. For the sake of completeness, we
present a self-contained solution.

\begin{restatable}{lemma}{lemSolveNoPipes}\label{lem:solve-no-pipes}
  An instance of \pqplan without pipes can be solved in $O(m)$ time.
  A~valid embedding can be computed in the same time, if it exists.
\end{restatable}
\newcommand{\proofSolveNoPipes}{
\begin{proof}
  We replace each Q-vertex by a wheel of the respective degree. Note that each
  such wheel is biconnected and entirely contained in its rigid triconnected
  component. We now use the decomposition in triconnected components to
  represent all possible planar embeddings. As the wheel-replacements yield an
  instance that is linear in the size of the initial instance, this
  decomposition can be done in $O(m)$ time. If at least one of the rigids has
  no planar embedding, abort and report a non-planar instance. It remains to
  restrict the possible embeddings of the rigids so that all Q-constraints are
  satisfied.

  We construct an instance of 2-SAT, where each solution corresponds to a
  planar embedding $\mathcal{E}$ that is a valid solution for $I$. For every
  Q-vertex $v$ the boolean variable $x_v$ is \texttt{true} if the rotation of
  $v$ in $\mathcal{E}$ is equal to the default rotation of $v$ (i.e. if
  $\mathcal{E}(v)=\psi(v)$).
Additionally, for every Q-constraint cell $Q\in\mathcal Q$ we add a boolean
  variable $x_Q$, and for every Q-vertex $v\in Q$ we add the constraint
  $(x_v\vee\neg x_Q)\wedge(\neg x_v\vee x_Q)$. This ensures that the rotations
  of the Q-vertices are consistent within their cell and thus satisfy the
  Q-constraints.
We still need to ensure that the 2-SAT instance allows only planar embeddings. For
  each rigid $\mu$, we fix one of its two planar embeddings as its default
  embedding $\mathcal E_\mu$ and add another boolean variable $x_\mu$,
  indicating whether $\mathcal E_\mu$ or $\overline{\mathcal E_\mu}$ shall be
  used in $\mathcal E$. Due to the wheel replacement, each Q-vertex~$v$ is
  entirely contained in its rigid $\mu$, which can be found in constant time
  using one of the incident edges. For every Q-vertex $v$ we now add one of
  the following two constraints with regard to its rigid $\mu$:
  either 1) $(x_v\vee\neg x_\mu)\wedge(\neg x_v\vee x_\mu)$ if $\psi(v)=\mathcal E_\mu(v)$, or
  2) $(x_v\vee x_\mu)\wedge(\neg x_v\vee\neg x_\mu)$ if $\psi(v)=\overline{\mathcal E_\mu(v)}$.
  This ensures that the rotation of every Q-vertex is consistent with the
  planar embedding of its rigid.

  In the resulting 2-SAT instance, we have a boolean variable for each
  Q-vertex, Q-constraint and rigid, and four constraints for each Q-vertex. The
  constructed 2-SAT formula thus has size in $O(m)$ and can be solved in linear time.
  If it has no solution, we report an invalid instance and abort. Otherwise,
  we can use $x_\mu$ to decide whether $\mathcal E_\mu$ should be mirrored or
  not. Choosing a planar embedding for each bond, i.e. a permutation of the
  parallel virtual edges between the two poles, this yields a planar
  embedding~$\mathcal{E}$ that is a valid solution for~$I$.
\end{proof}}
\ifthenelse{\boolean{long}}{\proofSolveNoPipes}{}

\subsection{Finding a Reduced Instance}

As mentioned above, we exhaustively apply the operations \contract,
\propagate, and \simplify. We claim that this algorithm terminates and
yields a reduced instance after a polynomial number of steps.  The key
idea is that the operations always
make progress by either reducing the number of pipes, or by splitting
pipes into pipes of smaller degree.  This suggests that, eventually, we
arrive at an instance without pipes.  However, there are two caveats. 
First, the encapsulation in the first step of \contract creates new
pipes and thus has the potential to undo progress.  Second, the smaller
pipes resulting from splitting a pipe with \propagate might cause
further growth of the instance, potentially causing a super-polynomial
number of steps.

We resolve both issues by using a more fine-grained measure of
progress in the form of a potential function.  To overcome the first
issue, we show that for each application of \contract, the progress
that is undone in the first step is outweighed by the progress made
through the following join in the second step.  Similarly, for
the second issue, we show that the sum of the parts is no bigger than
the whole when splitting pipes.

As P-vertices of degree 3 or less are converted to Q-vertices (see
\Cref{deg-lt-3-to-q}), we use $\vdegs(u)=\vdegs(v)=\vdegs(\rho)=
\max\{\vdeg(x)-3,0\}$ to denote the number of incident edges that keep a
P-vertex $u$ (and also the other endpoint $v$ of its pipe
$\rho=(u,v,\varphi_{uv})$) from becoming converted to a Q-vertex. We
also partition the set of all pipes $\mathcal{P}$ into the two cells
$\mathcal{P}_{CC}$ and 
$\mathcal{P}_B=\mathcal{P}\setminus\mathcal{P}_{CC}$, where
$\mathcal{P}_{CC}$ contains all pipes where both endpoints are
cut-vertices.
We define the \emph{potential} of an instance $I$ as
  $\Phi(I) = \sum_{\rho\in\mathcal{P}_B} \vdegs(\rho) +
  \sum_{\rho\in\mathcal{P}_{CC}} (2\vdegs(\rho)-1)$.
\ifthenelse{\boolean{long}}{
  
  We show that the operations always decrease this potential. To analyze
  the potential change of \propagate and \contract, we need the following
  technical lemma for bounding the sum of the degrees of multiple smaller
  pipes replacing a single bigger pipe.
}{
  We show that the operations always decrease this potential. 
}

\newcommand{\lemSums}{
\begin{lemma}\label{lem:sums}
  Let $k\geq2, \,d_1\geq d_2\geq\ldots\geq d_k\geq1$ and $c\geq0$ be integers.
  Let $j=\max\left\{i\mid d_i\geq3\right\}$ and let $\ell=\left|\left\{i\mid d_i=2\right\}\right|$. If $3 \leq c + \ell + 2j$, then
  $k + \sum_{i=1}^{k} \max\{d_i-3, 0\} \leq c - 3 + \sum_{i=1}^{k} d_i$.
  If $4 \leq \sum_{i=1}^{k} d_i$, then
  $    \sum_{i=1}^{k} \max\{d_i-3, 0\} \leq   - 4 + \sum_{i=1}^{k} d_i$.
\end{lemma}
\begin{proof}
  Observe that the $d_i$ are ordered increasing and thus 
  $d_i \ge 3$ for $i=1,\dots,j$ and $d_i < 3$ for $i=j+1,\dots,k$.
  More specifically, it is $d_{j+1},\ldots,d_{j+\ell}=2$ and 
  $d_{j+\ell+1},\ldots,d_{k}=1$.
  This yields
  $\sum_{i=j+1}^{j+\ell} d_i = 2\cdot\ell$ and
  $\sum_{i=j+\ell+1}^{k} d_i = k-\ell-j$
  and we can also avoid the ``$\max$'' using
  $\sum_{i=1}^{k} \max\{d_i-3, 0\}
  = \sum_{i=1}^{j} (d_i-3) = -3j + \sum_{i=1}^{j} d_i$.
We now start at $3 \leq c + \ell + 2j$, which can be rewritten as
  $k -3j \leq c - 3 + 2\ell + (k-\ell-j)$.
  Adding $\sum_{i=1}^{j} d_i$ on both sides and using the above
  observations yields
  \[k + \sum_{i=1}^{k} \max\{d_i-3, 0\}
  = k -3j + \sum_{i=1}^{j} d_i 
  \leq c - 3 + \sum_{i=j+1}^{j+\ell} d_i
   + \sum_{i=j+\ell+1}^{k} d_i + \sum_{i=1}^{j} d_i
  = c - 3 + \sum_{i=1}^{k} d_i.\]

  Note that inserting $c=k-1$ in the first formula yields the second
  formula. It remains to show that in this case 
  $3 \leq c + \ell + 2j$ or the equivalent $4 \leq k + \ell + 2j$
  follow from $4 \leq \sum_{i=1}^{k} d_i$.
  If~$k\geq4$, the inequality obviously always holds.
  If $k=3$, it must be $j>1$ or $\ell>1$ as the sum is at least 4.
  If $k=2$, it must be $j>1$ or $\ell\geq2$ as the sum is at least 4.
\end{proof}
}
\ifthenelse{\boolean{long}}{\lemSums}{}

\begin{restatable}{lemma}{lemItPotDec}\label{lem:iteration-potential-decrease}
For an instance \pqplanvars of \pqplan and an instance \pqplanvars[']
that results from application of either \contract, \propagate or
\simplify to $I$, the following three properties hold:
\begin{compactenum}[(i)]
\item The potential reduction $\Delta\Phi = \Phi(I)-\Phi(I')$ is at
  least 1.
\item The number of nodes added to the graph satisfies
  $\Delta V = |V(G')|-|V(G)| \leq 2 \cdot \Delta \Phi + 12$.
\item If the operation replaces a connected component $C$ by one or
  multiple connected components, then each such component $C'$
  satisfies $\Delta E(C) = |E(C')|-|E(C)| \leq 2 \cdot \Delta \Phi$.
\end{compactenum}
\end{restatable}
\newcommand{\oppoteff}[1]{}
\newcommand{\proofItPotDec}{
\begin{proof}
We now analyze the effects of \contract, \propagate, and \simplify on the
measures $\Delta\Phi$, $\Delta V$ and $\Delta E(C)$ and show that the found
changes satisfy the claimed bounds.

\oppoteff{\contractTitle $\Delta\Phi$}
Operation \contractop in the first step encapsulates both cut-vertices
$u,v$ to their own star components. For each block incident to $u$,
this introduces two new vertices that are connected by a new pipe. Let
$d_1\geq d_2\geq\ldots\geq d_{k-1}\geq d_k\geq1$ be the degrees of the
$k\geq2$ ray vertices of $u$ after the encapsulation. As one end of the added
pipes is a block-vertex, the potential is increased by $\sum_{i=1}^k
\max\{d_i-3, 0\}$. The rays around $v$ increase the potential and
number of vertices likewise, where $d_1'\geq d_2'\geq\ldots\geq
d_{k'-1}'\geq d_{k'}'\geq1$ are the degrees of the $k'\geq2$ rays of
$v$ after the encapsulation. Using $\sum_{i=1}^{k}d_i
=\sum_{i=1}^{k'}d_i'=D$ it is $\vdeg(\rho)=\vdeg(v)=\vdeg(u)=D$ and
$\vdegs(\rho)=\vdegs(v)=\vdegs(u)=\max\{D-3,0\}=D-3$ as $u$ and $v$
are P-vertices of the same degree greater then three. In the second
step, removing $\rho$ connecting two cut-vertices together with its
endpoints reduces the potential by $2\vdegs(\rho)-1$ and we thus get
$\Delta\Phi\ =\ 2\cdot\left( D - 3 \right) - 1\ -
\sum_{i=1}^{k} \max\{d_i-3, 0\}\ - \sum_{i=1}^{k'} \max\{d_i'-3, 0\}$.

As $D\geq4$, we know from the second formula of \Cref{lem:sums} that
$\sum_{i=1}^{k} \max\{d_i-3, 0\} \leq (\sum_{i=1}^{k} d_i)-4 = D-4$
and also $\sum_{i=1}^{k'} \max\{d_i'-3, 0\} \leq D-4$. Using this
inequality in the formula above yields $\Delta\Phi = 2D-7 -
\sum_{i=1}^{k}\max\{d_i-3,0\} - \sum_{i=1}^{k'}\max\{d_i'-3,0\} \geq 
2D-7-(D-4)-(D-4)=1$ as claimed by (i).

\oppoteff{\contractTitle $\Delta V$}
As the encapsulation generates two vertices for each ray and the 
join removes two vertices, we have $\Delta V = 2k+2k'-2$.
\Cref{lem:sums} with $c=3$ unconditionally yields $k\leq
D-\sum_{i=1}^{k}\max\{d_i-3,0\}$ and $k'\leq
D-\sum_{i=1}^{k'}\max\{d_i'-3,0\}$. Now claim (ii) holds as
\[\Delta V \leq 2\cdot\left(D-\sum_{i=1}^{k}\max\{d_i-3,0\}
+D-\sum_{i=1}^{k'}\max\{d_i'-3,0\}-1\right)
=2\cdot\left(\Delta\Phi+6\right).\]

\oppoteff{\contractTitle $\Delta E(C)$}
In the first step of \contract, two new components with
$\vdeg(u)=\vdeg(v)$ edges each are added, which are then pairwise
combined in the second step, yielding a new component with
$\sum_{i=1}^{k} d_i = \sum_{i=1}^{k'} d_i'$ edges, which is no bigger
than the components of $u$ or $v$ as required for (iii).

\oppoteff{\propagateTitle $\Delta\Phi$}
Operation \propagateop replaces the pipe $\rho$ having block-vertex
$u$ as one endpoint by one pipe for each inner P-node of the non-trivial
embedding tree $\mathcal{T}_u$ of $u$.  If $\mathcal{T}_u$ only consists
of a single Q-node, we are removing a pipe without adding any new pipes,
nodes, or edges, thus properties (i)$-$(iii) are trivially satisfied.
Otherwise, let $d_1\geq d_2\geq\ldots\geq d_{k-1}\geq d_k$ be the
degrees of the $k\geq2$ inner vertices of $\mathcal{T}_u$.
The tree $\mathcal{T}_u$ has $\vdeg(u)$ leaves and
thus  $(\sum_{i=1}^k d_i) + \vdeg(u) = 2\cdot|E(\mathcal{T}_u)| =
2\cdot(|V(\mathcal{T}_u)|-1) = 2\cdot(k+\vdeg(u)-1)$ or, equivalently,
$\vdeg(u) = \left( \sum_{i=1}^k (d_i-2) \right) + 2$. As $u$ is a
P-vertex with degree at least 4, $\vdegs(u) = \vdeg(u) - 3 = \left(
\sum_{i=1}^k (d_i-2) \right) - 1$. As $\mathcal{T}_u$ contains no
vertices of degree 2, it is $d_i\geq3$ for $i=1,\ldots,k$. As the
added pipes all have one endpoint in the biconnected component of $u$,
the potential is reduced by $\Delta\Phi = \vdegs(u) - \sum_{i=1}^k
\max\{d_i-3,0\} = \left( \sum_{i=1}^k (d_i-2) \right) - 1 - \left(
\sum_{i=1}^k (d_i-3) \right) = k - 1 \geq 1$, which shows (i).

\oppoteff{\propagateTitle $\Delta V$}
Moreover, we replace the two endpoints of $\rho$ each by the inner
nodes of $\mathcal{T}_u$, yielding $\Delta V = 2k-2$ additional nodes.
Note that $\Delta V = 2k-2 = 2\cdot\Delta\Phi < 2\cdot\Delta\Phi+12$
as claimed by (ii).

\oppoteff{\propagateTitle $\Delta E(C)$}
As each inner node except for the root has one edge connecting it to
its parent, we also add $\Delta E(C) = k-1$ additional edges to each
component. Observe that $\Delta E(C) = k-1 < 2k-2 = 2\cdot\Delta\Phi$
as claimed by (iii).

\oppoteff{\simplifyTitle}
Operation \simplify always removes at least one pipe
$\rho\in\mathcal{P}_B$ and thus decreases the potential by at least
$\vdegs(\rho)$. If two pipes $\rho,\rho'$ are replaced by their
transitive shortcut (i.e.\ \Cref{case:simplify-matched} of \simplify
applies), this adds a new pipe $\rho^*$. If at least one endpoint
of $\rho^*$ is a block-vertex, the potential change is
$\Delta\Phi=2\vdegs(\rho)-\vdegs(\rho)$. Otherwise, both endpoints are
cut-vertices and $\rho^*$ belongs to $\mathcal{P}_{CC}$, yielding a
potential change of $\Delta\Phi=2\vdegs(\rho)-(2\vdegs(\rho)-1)=1$.
As no vertices or edges are added or removed $\Delta V=\Delta E(C)=0$,
which is less than $\Delta \Phi$.
\end{proof}
}
\ifthenelse{\boolean{long}}{\proofItPotDec}{
\begin{proof}[Proof Sketch.]
  We now analyze the effects of \contract on these three measures.
  The operations \propagate and \simplify are discussed in the full version.

  Operation \contractop in the first step encapsulates both cut-vertices $u,v$ to their own star components.
  For each block incident to $u$, this introduces two new vertices that are connected by a new pipe.
  Let $d_1,\dots,d_k$ be the degrees of the $k\geq2$ ray vertices of $u$ after the encapsulation.
  As one end of the added pipes is a block-vertex, the potential is increased by $\sum_{i=1}^k \max\{d_i-3, 0\}$.
  Likewise, the pipes of the $k'$ rays with degrees~$d_1',\dots,d_{k'}'$ around $v$ increase the potential by $\sum_{i=1}^{k'} \max\{d_i'-3, 0\}$.
  Using $\vdeg(\rho)=\sum_{i=1}^{k}d_i =\sum_{i=1}^{k'}d_i'=D\geq3$
  it is $\vdegs(\rho)=\max\{D-3,0\}=D-3$.
  In the second step, removing $\rho$ connecting two cut-vertices together with its endpoints reduces the potential by $2\vdegs(\rho)-1$ and we thus get
  $\Delta\Phi\ =\ 2\cdot\left( D - 3 \right) - 1\ -
  \sum_{i=1}^{k} \max\{d_i-3, 0\}\ - \sum_{i=1}^{k'} \max\{d_i'-3, 0\}$.
In the full version, we show that this
  yields $\Delta\Phi \geq 1$ as claimed by (i).
  As the encapsulation generates two vertices for each ray and the 
  join removes two vertices, we have $\Delta V = 2k+2k'-2$.
  We also show that claim (ii) holds as
  $\Delta V \leq 2\cdot\left(\Delta\Phi+6\right)$.
In the first step of \contract, two new components with
  $\vdeg(u)=\vdeg(v)$ edges each are added, which are then
  combined in the second step, yielding a new component with
  $\sum_{i=1}^{k} d_i$ edges.
  This is no larger than the components of $u$ or $v$ as required for (iii).
\end{proof}
}

\newcommand{\lemMaxPot}{
\begin{lemma}\label{lem:max-potential}
Let $I$ be an instance of \pqplan. Then, $\Phi(I) < 2m$.
\end{lemma}
\begin{proof}
  Each pipe $\rho$ matching two vertices $u$ and~$v$ contributes at
  most $2\vdegs(\rho)-1 < \deg(u) + \deg(v)$ to the potential.  Since
  each vertex is part of at most one pipe, the sum of all potentials
  is bounded by~$\sum_{v \in V} \deg(v) = 2m$.
\end{proof}
}

\newcommand{\proofRuntime}{
\begin{proof}
  By \Cref{lem:iteration-potential-decrease} the potential function
  decreases with each applied operation.  Therefore, by
  \Cref{lem:max-potential}, after $k\leq2m$ operations a reduced
  instance $\pqplanvars[']$ is reached.  We claim that the resulting
  graph $G'=(V',E')$ has $|V'| \leq |V| + 4 \cdot |E| + 12 \cdot k$
  vertices and each connected component $C'$ of $G'$ has $|E(C')| \leq 5
  \cdot |E|$ edges.

  Let $\Delta\Phi_i$ for $i\in[1,\ldots,k]$ be the potential reduction
  caused by the $i$th applied operation. According to
  \Cref{lem:iteration-potential-decrease}, this operation also added
  $\Delta V_i \leq 2 \cdot \Delta\Phi_i + 12$ vertices to the graph. By
  \Cref{lem:max-potential} it is $\sum_{i=1}^k \Delta\Phi_i \leq \Phi(I)
  < 2|E|$ and thus $|V'| = |V| + \sum_{i=1}^k \Delta V_i \leq |V| + 4
  \cdot |E| + 12 \cdot k$. Additionally, if the $i$th operation replaces
  a connected component~$C_{i-1}$ by one or multiple connected
  components, then each such component $C_{i}$ satisfies
  $|E(C_{i})|-|E(C_{i-1})| \leq 2 \cdot \Delta \Phi_i$. Each connected
  component $C_1$ of the initial graph $G$ has at most $|E|$ edges.
  Using the same argument as above, we obtain $|E(C')| = |E(C_k)| \leq
  |E| + 4 \cdot |E|$.

  As $m\in \Omega(n)$, this shows that the resulting instance has $O(m)$
  vertices and each connected component has $O(m)$ edges.  Computing the
  embedding trees for a single connected component on demand can thus be
  done in $O(m)$ time.  In addition to this computation, each of the $k$
  operations takes time linear in the degree of the vertex it is applied
  to, which also is in $O(m)$.  Thus, each operation takes $O(m)$ time
  and, in total, it takes $O(m^2)$ time to reach a reduced instance.  As
  the size of this reduced instance is also in $O(m^2)$, using
  \Cref{lem:solve-no-pipes} for finding a solution for the reduced
  instance can be done in $O(m^2)$ time.
\end{proof}
}

\ifthenelse{\boolean{long}}{
  With this lemma, we know that each step decreases the potential by at least 1 without growing the graph too much.
  The following shows an upper bound on the potential.

  \lemMaxPot

  This can be used to bound the size of instances resulting from applying
  multiple operations consecutively and finally to bound the time required
  to find a solution for an instance.
  
  \begin{restatable}{theorem}{thmRuntime}\label{thm:runtime}
    \pqplan can be solved in $O(m^2)$ time.
  \end{restatable}
  \proofRuntime
}{
  With this lemma, we know that each step decreases the potential by at least 1 without growing the graph too much.
  As each vertex contributes at most twice its degree, initially~$\Phi(I) \le 2m$.
This can then be used to bound the size of instances resulting from applying
  multiple operations consecutively and finally to bound the time required
  to find a solution for an instance.
  
  \begin{restatable}{theorem}{thmRuntime}\label{thm:runtime}
    \pqplan can be solved in $O(m^2)$ time.
  \end{restatable}
  \begin{proof}[Proof Sketch.]
    By \Cref{lem:iteration-potential-decrease} the potential function
    decreases with each applied operation.  As initially $\Phi(I) \le 2m$, a reduced
    instance $I'$ is reached after $ k \leq 2m$ operations.
    It can be shown that each connected
    component of $I'$ has $O(m)$ edges, allowing an embedding tree to be
    computed in $O(m)$ time. Each of the $k$ operations runs in linear
    time once the PQ-tree it requires is available. In total, it thus
    takes $O(m^2)$ time to reach a reduced instance. As its size is also
    in $O(m^2)$, \Cref{lem:solve-no-pipes} can be applied to find a
    solution in $O(m^2)$ time.
  \end{proof}
}

\ifthenelse{\boolean{long}}{}{
  \section{Conclusion}

  We have given a quadratic-time algorithm for \pqplan, which improves
  the previous~$O(n^8)$-time algorithm for the linear-time
  equivalent problem \atemb~\cite{Fulek2019}.  Similar to Goldberg and
  Tarjan's push-relabel algorithm, it relies on few and simple
  operations that can be applied pretty much in an arbitrary order.
  This highlights where and how progress is made and thereby clearly
  exposes key ideas that also underlie the algorithm for \atemb.
  For a direct comparison of the approaches, we refer to the full version.

  The applications of \pqplan include solving \cplan, \atemb, \consefe and \ppqconp in quadratic time,
  thanks to linear-time reductions to \pqplan for all of them described in the full version.
  This improves over
  the previously fastest algorithms via the $O(n^8)$-time algorithm
  for \atemb.  In the case of \consefe the reduction used
  in~\cite{Fulek2019} includes a quadratic blowup and therefore yields
  an $O(n^{16})$-algorithm.  Our direct linear-time reduction
  leads to a quadratic algorithm.  

  \bibliography{bibliography}
}

\section{Applications}\label{sec:applications}\label{sec:atemb}
We start our discussion of applications of \pqplan with the \atemb problem.
Recall from the introduction that \atemb has two graphs as input.  One
graph represents a molecule structure with atoms and pipes between them,
the other graph is mapped onto that structure such that edges connect
vertices on a single atom or vertices on neighboring atoms through the
corresponding pipe.  As observed by Fulek and Tóth~\cite[Observation
1]{Fulek2019}, \atemb can be equivalently viewed as follows.  For each
atom consider the graph on that atom together with, for each incident
pipe, one virtual vertex that is incident to all edges that would
normally go through this pipe to a neighboring atom. Note that each pipe
has two virtual vertices corresponding to it, one on each of its
incident atoms.  Then an instance of \atemb is positive if and only
if all these graphs can be embedded such that every pair of virtual
vertices corresponding to the same pipe have opposite rotation.
This directly reduces \atemb to \pqplan.

\begin{theorem}
\atemb can be solved in $O(m^2)$ time.
\end{theorem}

To reduce \cplan to \pqplan, we use the CD-tree~\cite{Blaesius2015};
see also Figure~\ref{fig:cdtree}.  Each node of the CD-tree
corresponds to a graph, called its \emph{skeleton}.  Some vertices of
a skeleton are \emph{virtual vertices}.  Each virtual vertex
corresponds to exactly one virtual vertex in a different skeleton,
called its \emph{twin}, and there is a bijection between the edges
incident to a virtual vertex and its twin.  The tree structure of the
CD-tree comes from these correspondences between twins, that is, the
CD-tree has an edge between two nodes if and only if their skeletons
have virtual vertices that are twins of each other.  It is known that
a clustered graph is c-planar if and only if the skeletons of all
nodes in its CD-tree can be embedded such that every virtual vertex
and its twin have opposite {rotation~\cite[Theorem~1]{Blaesius2015}\nolinebreak\punctuationfootnote{The theorem
  originally requires ``the same'' instead of ``opposite'' rotations,
  which is equivalent due to the tree structure.}.}  As the CD-tree has
linear size and can be computed in linear time, this yields a linear
reduction from \cplan to \pqplan.

\begin{figure}[t]
  \centering
  \begin{subfigure}[t]{0.4\textwidth}
    \centering
    \includegraphics[scale=1,page=1]{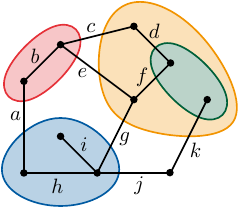}
    \caption{}
  \end{subfigure}\begin{subfigure}[t]{0.6\textwidth}
    \centering
    \includegraphics[scale=1,page=2]{graphics/cdtree.pdf}
    \caption{}
  \end{subfigure}\caption{An instance of \cplan (a) and its CD-tree representation (b),
    where each skeleton is shown with a gray background and the virtual
    vertices are shown as colored disks.}
  \label{fig:cdtree}
\end{figure}

\begin{theorem}
\cplan can be solved in $O((n+d)^2)$ time, where $d$ is the number of cluster boundary crossings.
\end{theorem}
\begin{proof}
We use the disjoint union of all skeletons of the CD-tree and match each
virtual vertex with its twin using a pipe. We can assume that the
underlying graph of a \cplan instance has no multi-edges and it must be
planar to be cluster-planar, thus its number of edges satisfies $m \in
O(n+d)$ and the runtime of our algorithm is $O((n+d)^2)$.
\end{proof}

Another problem that investigates the enforcement of rotation
constraints in planar embeddings is \ppqconp \cite{Blaesius2011}. Here,
each vertex in the graph can be annotated with a PQ-tree that limits the
rotations of (some of) its incident edges.

\begin{figure}[t]
  \centering
  \includegraphics[scale=1,page=3]{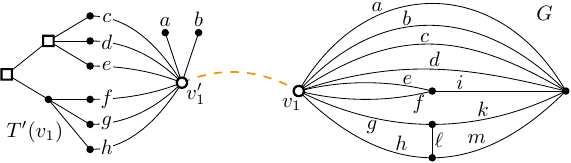}
  \caption{The equivalent \pqplan instance for an instance of
    \ppqconp, where the PQ-tree $T'(v_1)$ (left) restricts the order of
    the edges $\{a,\ldots,h\}$ around the vertex $v_1$ in $G$ (right).
    The cap-vertex $v_1'$ was added together with two degree-1 vertices.}
  \label{fig:partial-pq-const}
\end{figure}

\begin{theorem}
\ppqconp can be solved in $O(m^2)$ time.
\end{theorem}
\begin{proof}
Instances of \ppqconp can be converted to equivalent instances of
\pqplan by first adding the PQ-trees to the graph, converting Q-nodes to
Q-vertices similarly to \propagate. Afterwards, a cap-vertex is added to
each PQ-tree and all leaves are connected to the respective cap-vertex. To
be able to match the cap-vertices with the vertices in the original graph,
further degree-1 vertices are connected to the cap-vertices until their
degree matches the respective node in the original graph; see
\Cref{fig:partial-pq-const}.  The pipe then ensures that the rotation of
the matched vertex is compatible with the PQ-tree it was annotated with.
Note that previous algorithms are geared towards biconnected instances
and cannot handle cut-vertices of degree more than 5, while this approach
works for general graphs.
\end{proof}

\newcommand*\circled[1]{ \protect\tikz[baseline=(char.base)]{ \protect\node[shape=circle,draw,inner sep=0.2pt,scale=0.7] (char) {#1};}}

For two graphs $G^{\circled{1}}$ and $G^{\circled{2}}$, \sefe is
equivalent to finding a pair of planar embeddings that induce the same
(i.e. consistent) cyclic edge orderings and the same (i.e. consistent)
relative positions on their common graph $G=G^{\circled{1}}\cap
G^{\circled{2}}$~\cite{Juenger2009}. Our algorithm can be used to
provide the synchronization for the first half of this requirement,
which is sufficient for instances with a connected shared graph.

\begin{figure}[t]
  \centering
  \begin{subfigure}[t]{0.45\textwidth}
    \centering
    \includegraphics[scale=1,page=5]{graphics/sefe.pdf}
    \includegraphics[scale=1,page=1]{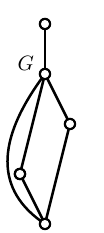}
    \includegraphics[scale=1,page=3]{graphics/sefe.pdf}
    \caption{}
  \end{subfigure}\begin{subfigure}[t]{0.55\textwidth}
    \centering
    \includegraphics[scale=1,page=8]{graphics/sefe.pdf}
    \caption{}
  \end{subfigure}\caption{
    An instance of \consefe $G=G^{\circled{1}}\cap G^{\circled{2}}$ (a)
    and the equivalent \pqplan instance $I$ (b).}
  \label{fig:sefe}
\end{figure}

\begin{theorem}
\consefe can be solved in $O(n^2)$ time.
\end{theorem}
\begin{proof}
We add both $G^{\circled{1}}$ and $G^{\circled{2}}$ to the \pqplan
instance and also add a bond $b^{\circled{1}}b^{\circled{2}}$ for each
vertex $x\in G$.  The parallel edges of the bond correspond to the edges
incident to $x$ in $G$.  Again, we add further degree-1 vertices so that
we can match $x^{\circled{1}}$ with $b^{\circled{1}}$ and
$x^{\circled{2}}$ with $b^{\circled{2}}$; see \Cref{fig:sefe}. The
\pqplan algorithm can then be used to obtain embeddings $\mathcal E_1$
and $\mathcal E_2$ for $G^{\circled{1}}$ and $G^{\circled{2}}$, respectively.
As pipes reverse the order of incident edges, a solution for the \pqplan
instance will have one of the two graphs mirrored with respect to the
graph shared with the other, that is, the solution for the SEFE instance is
$\mathcal E_1, \overline{\mathcal E_2}$.
\end{proof}

\let\ftop\textsf
\section{Comparison with the Fulek-Tóth Algorithm}\label{sec:comparison}

We want to point out that there are many parallels between our algorithm for solving \pqplan and Fulek and Tóth's solution to \atemb.
First note that both problems are linear-time equivalent, a reduction in the one direction has been given in the previous section.
To reduce from \pqplan to \atemb, all connected components can be turned into atoms, subdividing pipes that loop back to the same component with a trivial atom.
It remains to encode the Q-constraints, which can be done by converting all Q-vertices to wheels and synchronizing the centers for all Q-vertices in the same partition with pipes to one additional atom per partition.
Using a triconnected graph for this atom ensures that exactly two flips are possible and all Q-vertices in the partition are flipped the same way.
\smallskip

\noindent
The algorithm of Fulek and Tóth relies on seven basic operations:
\begin{compactitem}
\item \ftop{Suppress} eliminates pipes with degree two or less,
\item \ftop{Delete} removes certain edges when all affected components are subcubic,
\item \ftop{Split} splits an atom into its connected components, and
\item \ftop{Detach} splits unmatched cut-vertices.
\item \ftop{Enclose} encloses a cut-vertex inside its own atom,
\item \ftop{Contract} joins two neighboring atoms (under certain conditions), and
\item \ftop{Stretch} encodes the fact that a subset of edges incident
  to a vertex must be consecutive by splitting the vertex into two.
\end{compactitem}
Fulek and Tóth carefully orchestrate these operations into two large subroutines,
which they then iteratively use to globally reduce the maximum degree of cutvertices that correspond to pipes.
Eventually all such cutvertices have small degree and the problem can be solved directly
as each part of the instance is either subcubic or what they call toroidal.

Our algorithm only uses four basic operations: \convert, \contract, \simplify, and \propagate.
These operations can be applied in an arbitrary order to reach a reduced instance, which can then be solved directly.

In direct comparison, our operation \convert handles all the cleanup (and more) that Fulek and Tóth achieve with \ftop{Suppress} and \ftop{Delete}.
Further, in our context, the operation \ftop{Split} is not needed, since we do not differentiate between connected components and atoms,
and \ftop{Detach} is not needed since we handle unmatched cutvertices in our base case algorithm rather than during the reduction phase.

Our operation \contract replaces multiple calls to \ftop{Enclose} and \ftop{Contract} as a much more targeted solution.
The former also encompasses the stretching of local branches in Carmesin's work~\cite{Carmesin2017}.
Our operation and the conditions under which it can (and must be) applied clearly expose the key insight that also stands behind the algorithm of Fulek and Tóth,
and which we have formalized in Lemma~\ref{lem:bipartite-all-embeddings-respect-cut}:
any planar embedding of the bipartite graph resulting from contracting two stars at their centers respects the cut made for splitting both stars.
Furthermore, there are parallels between \ftop{Contract} and case \ref{case:simplify-matched} of \simplify (i.e. when a bond links two distinct pipes)\punctuationfootnote{
\ftop{Contract} could also be used instead of \propagate between two block-vertices of distinct connected components, although this would break our runtime analysis.}.
Finally, thanks to the use of PQ- and SPQR-trees, \propagate can be seen as a much more focused replacement of multiple iterations of \ftop{Stretch}.
An example for how the operations relate can be seen in Figure 11 of the paper by Fulek and Tóth~\cite{Fulek2019},
where an instance on which \simplify and then \propagate can (and clearly should) be applied,
is first \ftop{Contract}ed twice in Step (iv.a) and then \ftop{Stretch}ed in (a possibly later iteration of) Step (v.c) of their Subroutine 2.

To conclude, our approach makes it clear to see where progress is made, uncovering important ideas that are not obvious in the global degree-reduction approach employed by Fulek and Tóth.

\ifthenelse{\boolean{long}}{
  \section{Conclusion}
  We have given a quadratic-time algorithm for \pqplan, which improves the previous~$O(n^8)$-time algorithm for the linear-time equivalent problem \atemb~\cite{Fulek2019}.
  Similar to Goldberg and Tarjan's push-relabel algorithm, it relies on few and simple operations that can be applied in an arbitrary order.
  They also highlight where and how progress is made and thereby clearly expose key ideas that also underlie the algorithm for \atemb.

  The applications of \pqplan include solving \cplan, \atemb, \consefe and \ppqconp in quadratic time, thanks to linear-time reductions to \pqplan for all of them.
  This improves over the previously fastest algorithms via the $O(n^8)$-time algorithm for \atemb.
  In the case of \consefe the reduction used in~\cite{Fulek2019} includes a quadratic blowup and therefore yields an $O(n^{16})$-algorithm.
  Our direct linear-time reduction leads to a quadratic algorithm.

  \bibliography{bibliography}

\begin{thebibliography}{10}

\bibitem{Akitaya2019}
H.~A. Akitaya, R.~Fulek, and C.~D. T{\'{o}}th.
\newblock Recognizing weak embeddings of graphs.
\newblock {\em {ACM} Transactions on Algorithms}, 15(4):1--27, 2019.
\newblock \href {https://doi.org/10.1145/3344549} {\path{doi:10.1145/3344549}}.

\bibitem{ad-scp-16}
P.~Angelini and G.~{Da Lozzo}.
\newblock {SEFE} = c-planarity?
\newblock {\em The Computer Journal}, 59(12):1831--1838, 2016.
\newblock \href {https://doi.org/10.1093/comjnl/bxw035}
  {\path{doi:10.1093/comjnl/bxw035}}.

\bibitem{Angelini2019}
P.~Angelini and G.~{Da Lozzo}.
\newblock Clustered planarity with pipes.
\newblock {\em Algorithmica}, 81(6):2484--2526, 2019.
\newblock \href {https://doi.org/10.1007/s00453-018-00541-w}
  {\path{doi:10.1007/s00453-018-00541-w}}.

\bibitem{Battista1996}
G.~D. Battista and R.~Tamassia.
\newblock On-line maintenance of triconnected components with {SPQR}-trees.
\newblock {\em Algorithmica}, 15(4):302--318, 1996.
\newblock \href {https://doi.org/10.1007/bf01961541}
  {\path{doi:10.1007/bf01961541}}.

\bibitem{Blaesius2017}
T.~Bl{\"{a}}sius, A.~Karrer, and I.~Rutter.
\newblock Simultaneous embedding: Edge orderings, relative positions,
  cutvertices.
\newblock {\em Algorithmica}, 80(4):1214--1277, 2017.
\newblock \href {https://doi.org/10.1007/s00453-017-0301-9}
  {\path{doi:10.1007/s00453-017-0301-9}}.

\bibitem{Blaesius2011}
T.~Bl{\"{a}}sius and I.~Rutter.
\newblock Simultaneous {PQ}-ordering with applications to constrained embedding
  problems.
\newblock {\em {ACM} Trans. Algorithms}, 12(2):16:1--16:46, 2016.
\newblock \href {https://doi.org/10.1145/2738054} {\path{doi:10.1145/2738054}}.

\bibitem{bkr-hgdv-13}
T.~Bläsius, S.~G. Kobourov, and I.~Rutter.
\newblock Simultaneous embedding of planar graphs.
\newblock In R.~Tamassia, editor, {\em Handbook of Graph Drawing and
  Visualization}, chapter~11, pages 349--381. CRC Press, Taylor \& Francis
  Group, 2013.
\newblock \href {http://arxiv.org/abs/1204.5853} {\path{arXiv:1204.5853}}.

\bibitem{Blaesius2015}
T.~Bläsius and I.~Rutter.
\newblock A new perspective on clustered planarity as a combinatorial embedding
  problem.
\newblock {\em Theoretical Computer Science}, 609:306--315, 2016.
\newblock \href {http://arxiv.org/abs/1506.05673v1}
  {\path{arXiv:1506.05673v1}}, \href
  {https://doi.org/10.1016/j.tcs.2015.10.011}
  {\path{doi:10.1016/j.tcs.2015.10.011}}.

\bibitem{Booth1976}
K.~S. Booth and G.~S. Lueker.
\newblock Testing for the consecutive ones property, interval graphs, and graph
  planarity using {PQ}-tree algorithms.
\newblock {\em Journal of Computer and System Sciences}, 13(3):335--379, 1976.
\newblock \href {https://doi.org/10.1016/s0022-0000(76)80045-1}
  {\path{doi:10.1016/s0022-0000(76)80045-1}}.

\bibitem{Carmesin2017}
J.~Carmesin.
\newblock Embedding simply connected 2-complexes in 3-space -- {V}. {A} refined
  {K}uratowski-type characterisation, 2017.
\newblock \href {http://arxiv.org/abs/1709.04659v3}
  {\path{arXiv:1709.04659v3}}.

\bibitem{Cortese2008}
P.~F. Cortese, G.~D. Battista, F.~Frati, M.~Patrignani, and M.~Pizzonia.
\newblock C-planarity of c-connected clustered graphs.
\newblock {\em Journal of Graph Algorithms and Applications}, 12(2):225--262,
  2008.
\newblock \href {https://doi.org/10.7155/jgaa.00165}
  {\path{doi:10.7155/jgaa.00165}}.

\bibitem{Cortese2018}
P.~F. Cortese and M.~Patrignani.
\newblock Clustered planarity = flat clustered planarity.
\newblock In T.~C. Biedl and A.~Kerren, editors, {\em Proceedings of the 26th
  International Symposium on Graph Drawing and Network Visualization (GD'18)},
  volume 11282 of {\em LNCS}, pages 23--38. Springer, 2018.
\newblock \href {https://doi.org/10.1007/978-3-030-04414-5_2}
  {\path{doi:10.1007/978-3-030-04414-5_2}}.

\bibitem{d-gt-17}
R.~Diestel.
\newblock {\em Graph Theory}.
\newblock Graduate Texts in Mathematics. Springer, 5th edition edition, 2017.
\newblock \href {https://doi.org/10.1007/978-3-662-53622-3}
  {\path{doi:10.1007/978-3-662-53622-3}}.

\bibitem{Feng1995}
Q.-W. Feng, R.~F. Cohen, and P.~Eades.
\newblock Planarity for clustered graphs.
\newblock In P.~G. Spirakis, editor, {\em Proceedings of the 3rd Annual
  European Symposium on Algorithms (ESA'95)}, volume 979 of {\em LNCS}, pages
  213--226. Springer, 1995.
\newblock \href {https://doi.org/10.1007/3-540-60313-1_145}
  {\path{doi:10.1007/3-540-60313-1_145}}.

\bibitem{Fulek2015}
R.~Fulek, J.~Kyn{\v{c}}l, I.~Malinovi{\'{c}}, and D.~P{\'{a}}lvölgyi.
\newblock Clustered planarity testing revisited.
\newblock {\em The Electronic Journal of Combinatorics}, 22(4), 2015.
\newblock \href {https://doi.org/10.37236/5002} {\path{doi:10.37236/5002}}.

\bibitem{Fulek2019}
R.~Fulek and C.~D. Tóth.
\newblock Atomic embeddability, clustered planarity, and thickenability.
\newblock In {\em Proceedings of the 31st Annual {ACM}-{SIAM} Symposium on
  Discrete Algorithms (SODA'20)}, pages 2876--2895. SIAM, 2020.
\newblock \href {http://arxiv.org/abs/1907.13086v1}
  {\path{arXiv:1907.13086v1}}, \href
  {https://doi.org/10.1137/1.9781611975994.175}
  {\path{doi:10.1137/1.9781611975994.175}}.

\bibitem{Gutwenger2002}
C.~Gutwenger, M.~J{\"{u}}nger, S.~Leipert, P.~Mutzel, M.~Percan, and
  R.~Weiskircher.
\newblock Advances in c-planarity testing of clustered graphs.
\newblock In S.~G. Kobourov and M.~T. Goodrich, editors, {\em Proceedings of
  the 10th International Symposium on Graph Drawing (GD'02)}, volume 2528 of
  {\em LNCS}, pages 220--235. Springer, 2002.
\newblock \href {https://doi.org/10.1007/3-540-36151-0_21}
  {\path{doi:10.1007/3-540-36151-0_21}}.

\bibitem{Carsten2008}
C.~Gutwenger, K.~Klein, and P.~Mutzel.
\newblock Planarity testing and optimal edge insertion with embedding
  constraints.
\newblock {\em Journal of Graph Algorithms and Applications}, 12(1):73--95,
  2008.
\newblock \href {https://doi.org/10.7155/jgaa.00160}
  {\path{doi:10.7155/jgaa.00160}}.

\bibitem{gm-ltisp-00}
C.~Gutwenger and P.~Mutzel.
\newblock A linear time implementation of {SPQR}-trees.
\newblock In J.~Marks, editor, {\em Proceedings of the 8th International
  Symposium on Graph Drawing (GD'00)}, volume 1984 of {\em LNCS}, pages 77--90.
  Springer, 2000.
\newblock \href {https://doi.org/10.1007/3-540-44541-2_8}
  {\path{doi:10.1007/3-540-44541-2_8}}.

\bibitem{Hadlock1975}
F.~Hadlock.
\newblock Finding a maximum cut of a planar graph in polynomial time.
\newblock {\em {SIAM} J. Comput.}, 4(3):221--225, 1975.
\newblock \href {https://doi.org/10.1137/0204019} {\path{doi:10.1137/0204019}}.

\bibitem{ht-dgtc-73}
J.~E. Hopcroft and R.~E. Tarjan.
\newblock Dividing a graph into triconnected components.
\newblock {\em {SIAM} J. Comput.}, 2(3):135--158, 1973.
\newblock \href {https://doi.org/10.1137/0202012} {\path{doi:10.1137/0202012}}.

\bibitem{Hsu2001}
W.-L. Hsu.
\newblock {PC}-trees vs. {PQ}-trees.
\newblock In J.~Wang, editor, {\em Proceedings of the 7th Annual International
  Conference on Computing and Combinatorics (COCOON'01)}, volume 2108 of {\em
  LNCS}, pages 207--217. Springer, 2001.
\newblock \href {https://doi.org/10.1007/3-540-44679-6_23}
  {\path{doi:10.1007/3-540-44679-6_23}}.

\bibitem{Juenger2009}
M.~J{\"{u}}nger and M.~Schulz.
\newblock Intersection graphs in simultaneous embedding with fixed edges.
\newblock {\em Journal of Graph Algorithms and Applications}, 13(2):205--218,
  2009.
\newblock \href {https://doi.org/10.7155/jgaa.00184}
  {\path{doi:10.7155/jgaa.00184}}.

\bibitem{Lengauer1989}
T.~Lengauer.
\newblock Hierarchical planarity testing algorithms.
\newblock {\em Journal of the ACM}, 36(3):474--509, 1989.
\newblock \href {https://doi.org/10.1145/65950.65952}
  {\path{doi:10.1145/65950.65952}}.

\bibitem{n-acmc-68}
L.~Neuwirth.
\newblock An algorithm for the construction of 3-manifolds from 2-complexes.
\newblock {\em Mathematical Proceedings of the Cambridge Philosophical
  Society}, 64(3):603--614, 1968.
\newblock \href {https://doi.org/10.1017/S0305004100043279}
  {\path{doi:10.1017/S0305004100043279}}.

\bibitem{Schaefer2013}
M.~Schaefer.
\newblock Toward a theory of planarity: Hanani-tutte and planarity variants.
\newblock {\em Journal of Graph Algorithms and Applications}, 17(4):367--440,
  2013.
\newblock \href {https://doi.org/10.7155/jgaa.00298}
  {\path{doi:10.7155/jgaa.00298}}.

\bibitem{Shih1999}
W.-K. Shih and W.-L. Hsu.
\newblock A new planarity test.
\newblock {\em Theoretical Computer Science}, 223(1-2):179--191, 1999.
\newblock \href {https://doi.org/10.1016/s0304-3975(98)00120-0}
  {\path{doi:10.1016/s0304-3975(98)00120-0}}.

\bibitem{Tamassia2013}
R.~Tamassia, editor.
\newblock {\em Handbook of Graph Drawing and Visualization}.
\newblock CRC Press, Taylor \& Francis Group, 2014.
\newblock \href {https://doi.org/10.1201/b15385} {\path{doi:10.1201/b15385}}.

\end{thebibliography}
}{}

\end{document}